\documentclass[%
reprint, 
superscriptaddress,
amsmath,amssymb,
aps,
pra,
]{revtex4-1}
\usepackage{overpic}
\usepackage{graphicx}
\usepackage{dcolumn}
\usepackage{bm}

\usepackage{enumitem}
\usepackage{braket}
\usepackage{xcolor}
\usepackage{csquotes}
\usepackage{verbatim}
\usepackage{bbold}
\usepackage{amsthm}
\newtheorem{theorem}{Theorem}[section]

\newtheorem{lemma}[theorem]{Lemma}
\newtheorem{observation}{Observation}[section]

\DeclareMathOperator*{\argmin}{argmin}

\newcommand{\tr}{\mathrm{tr}}
\usepackage{hyperref}

\begin{document}

\preprint{APS/123-QED}

\title{Parallel ergotropy: Maximum work extraction via parallel local unitary operations}%

\author{Riccardo Castellano}
\email{riccardo.castellano@sns.it}
\affiliation{Scuola Normale Superiore, I-56126 Pisa, Italy}
\affiliation{ICFO - Institut de Ci\`encies Fot\`oniques, The Barcelona Institute of Science and Technology, 08860 Castelldefels, Barcelona, Spain}
\affiliation{Dipartimento di Fisica dell’Universit\`a di Pisa, Largo Pontecorvo 3, I-56127 Pisa, Italy}
\author{Ranieri Nery}
\affiliation{ICFO - Institut de Ci\`encies Fot\`oniques, The Barcelona Institute of Science and Technology, 08860 Castelldefels, Barcelona, Spain}
\author{Kyrylo Simonov}
\affiliation{Fakult\"at f\"ur Mathematik, Universit\"at Wien, Oskar-Morgenstern-Platz 1, 1090 Vienna, Austria}
\affiliation{ICFO - Institut de Ci\`encies Fot\`oniques, The Barcelona Institute of Science and Technology, 08860 Castelldefels, Barcelona, Spain}
\author{Donato Farina}
\email{donato.farina@unina.it}
\affiliation{Physics Department E. Pancini- Universit\`a degli Studi di Napoli Federico II, Complesso Universitario Monte S. Angelo- Via Cintia- I-80126 Napoli, Italy}
\affiliation{ICFO - Institut de Ci\`encies Fot\`oniques, The Barcelona Institute of Science and Technology, 08860 Castelldefels, Barcelona, Spain}

\date{\today}

\begin{abstract}
Maximum quantum work extraction is generally defined in terms of the \textit{ergotropy} functional, no matter how experimentally complicated is the implementation of the optimal unitary allowing for it, especially in the case of multipartite systems.
In this framework, 
we consider a quantum battery made up of many interacting sub-systems and study the maximum extractable work  via concurrent local unitary operations on each subsystem.
We call the resulting functional \textit{parallel ergotropy}.
Focusing on the bipartite case, 
we first observe that
parallel ergotropy outperforms
work extraction via \textit{egoistic} strategies,
in which
the first agent A extracts locally on its part the maximum available work and the second agent B, subsequently, extracts what is left on the other part.
For the agents, this showcases the need of cooperating for an overall benefit.
Secondly, from the informational point of view, we observe that the parallel capacity of a state can detect entanglement and compare it with the statistical entanglement witness that exploits fluctuations of stochastic work extraction.
Additionally, we face the technical problem of computing parallel ergotropy. 
We derive analytical upper bounds for specific classes of states and Hamiltonians and provide  receipts to obtain numerical upper bounds
via semi-definite programming in the generic case.
Finally, extending the concept of parallel ergotropy, we demonstrate that 
system's free-time evolution and application of local unitaries
allow one to saturate the gap with the ergotropy of the whole system.
\end{abstract}

\maketitle

\section{Introduction}

In the realm of classical thermodynamics, storage and transformation of energy embody its very essence. Rapid development of quantum technologies, however, requires revisiting foundational thermodynamic principles in order to address them at a microscopic level. In particular, counter-intuitive phenomena such as quantum superposition and entanglement offer many-body quantum systems as energy storage devices that overcome limitations of classical batteries \cite{Binder2015, Campaioli2023, Catalano2023}. Indeed, such quantum batteries are characterized by increasing charging power and energy storage capacity. The usefulness of a quantum battery can be captured by two quantities: 1) \textit{ergotropy}, the maximal amount of work that can be extracted from its state via cyclic unitary dynamics generated by the Hamiltonian of the quantum battery \cite{allahverdyan2004maximal}, and 2) \textit{quantum battery capacity}, the maximal amount of work that can be transferred via such unitary cycles \cite{BatteryCap}.

Experimental implementation of a quantum battery composed of several elements can be demanding since the resulting unitary operations may be highly non-local. This limitation can be faced by focusing on a subsystem instead of the entire compound, hence, addressing ergotropy with respect to unitary operations on this subsystem~\cite{LErgo}. However, it raises the following question: How much work can be extracted from the whole compound via concurrent local work extracting unitaries with respect to each subsystem? In this paper, we explore it by introducing a more experiment-friendly notion of \textit{parallel ergotropy} that captures concurrent work extraction from each subsystem composing a quantum battery. 
{We notice that this differs, e.g., from the ergotropy extraction protocol in multipartite systems recently explored in \cite{BipartiteQ.B.SolidState}, where one first switches off the interaction, requiring a form of non-local control over the whole system.}
We shall provide general lower and upper bounds for parallel ergotropy and calculate it explicitly for relevant specific classes of quantum states. 
In a similar manner, we introduce and study \textit{parallel capacity}, i.e., quantum battery capacity under concurrent local unitaries. We show that  parallel capacity can serve as an entanglement witness. Finally, we prove that combining free time evolutions and local unitaries on the subsystems, we can always close the gap with the \textit{global} ergotropy of the system.

The paper is organized as follows. In Sec.~\ref{sec:parallel-def}, we set the necessary framework and introduce the ergotropy, local ergotropy, and quantum battery capacity functionals. In~\ref{sec:parallel-def}, we introduce the figure of merit of the paper, the parallel ergotropy and parallel capacity functionals. 
In Sec.~\ref{sec:bounds} we establish general lower and upper bounds for this quantity. In Sec.~\ref{sec:LMM}, we calculate parallel ergotropy for locally maximally mixed states, 
or for arbitrary states but with null local Hamiltonian terms,
providing an analytic upper bound for quantum systems of arbitrary dimension and the exact value for two-qubit systems. 
In Sec.~\ref{sec:SDP} we introduce a general numerical technique for upper bounding parallel ergotropy and we perform a comparison among different estimates for two-qubit and two-qutrit systems. In Sec.~\ref{sec:ParallelCap}, we demonstrate that the parallel capacity functional can be used as an entanglement witness. A summary of the results together with an outlook to possible future developments are presented in Sec.~\ref{sec:conc}.

\section{Preliminaries}
\label{sec:preliminaries}
We start by recalling the definition of ergotropy. Let us consider a $d$-dimensional quantum system $S$ prepared in a state $\rho \in D(\mathcal{H})$, where $D(\mathcal{H})$ denotes the space of density operators on the Hilbert space $\mathcal{H}$ of $S$, and equipped with a Hamiltonian $H$. Under cyclic control of the Hamiltonian the state $\rho$ evolves due to a certain unitary operation $U$ inducing a change 
\begin{equation}\label{eq:MeanEn}
 \tr[\rho H] \rightarrow \tr[U \rho U^\dag H]
\end{equation}
in the mean energy of $S$. The latter can be associated with extraction of work from $\rho$, and the maximal amount of work that can be extracted via unitary operations is defined by the \textit{ergotropy functional}
\begin{eqnarray}
\nonumber \mathcal{E}(\rho, H) &:=& \tr[\rho H] - \min_{U \in \mathtt{U}(d)}\Bigl( \tr[ U \rho U^\dag H ] \Bigr)\\
&=& \tr\Bigl[ (\rho - \overline{U} \rho \overline{U}^\dag) H\Bigr]\,,
\label{eq:ergo}
\end{eqnarray}
where $\mathtt{U}(d)$ is the unitary group of degree $d$, and $\bar{U}$ is the optimal unitary operation, whose closed form can be provided in terms of the eigenstates of $\rho$ and $H$ \cite{allahverdyan2004maximal, Francica2017, Francica2020}. Therefore, \eqref{eq:ergo} establishes an upper bound on energy that can be removed from $S$ in the state $\rho$. Substituting minimization with maximization in \eqref{eq:ergo} provides a lower bound on change \eqref{eq:MeanEn} in the mean energy of $S$ characterizing how much it can be charged in the state $\rho$. The difference between these bounds characterizes, hence, the maximal amount of work that $S$ can transfer as a quantum battery. It can be formalized as \textit{quantum battery capacity functional}
\begin{eqnarray}
    \nonumber \mathcal{C}(\rho, H) &:=& \max_{U \in \mathtt{U}(d) }\Bigl( \tr [U \rho U^\dag H ] \Bigr) - \min_{U \in \mathtt{U}(d) }\Bigl( \tr [U \rho U^\dag H ] \Bigr) \\
    &=& \tr \Bigl[ (\underline{U} \rho \underline{U}^\dag - \overline{U} \rho \overline{U}^\dag ) H \Bigr] \,,
\end{eqnarray}
where $\underline{U}$ and $\overline{U}$ are the optimal unitaries achieving upper and lower bounds on change \eqref{eq:MeanEn} in the mean energy of $S$, respectively, with $\overline{U}$ being defined in \eqref{eq:ergo}.

Experimental implementation of work extracting unitary operation $\overline{U}$ (as well as work injecting $\underline{U}$) can be challenging, for example, if $S$ is multipartite, so that $\overline{U}$ can result to be highly non-local. Such limitations for global work extraction have been recently explored in Ref.~\cite{LErgo} and motivate to consider local work extraction, where the unitary cycle is performed on a subsystem of $S$. In particular, let us consider a bipartite system $S$ shared between two agents $A$ and $B$. It is equipped with a Hamiltonian $H$ and initially prepared in a state $\rho \in D(\mathcal{H}_a \otimes \mathcal{H}_b)$, where $\mathcal{H}_a$ and $\mathcal{H}_b$ are the Hilbert spaces of subsystems of $S$ associated to $A$ and $B$, respectively. Then the maximal amount of work that $A$ can extract from $S$ is given by the \textit{local ergotropy functional on $A$}
\cite{LErgo}
\begin{eqnarray}
\nonumber \mathcal{E}_{\rm La}(\rho, H) &:=& \max_{U_{a} \in \mathtt{U}(d_a)}\Bigl( \tr\Bigl[ \bigl( \rho - (U_{a} \otimes \mathbb{1}_{b})\rho (U^\dag_{a} \otimes \mathbb{1}_{b})\bigr) H \Bigr] \Bigr) \\
&=& \tr\Bigl[ \bigl( \rho - (\overline{U}_{a} \otimes \mathbb{1}_{b})\rho (\overline{U}^\dag_{a} \otimes \mathbb{1}_{b})\bigr) H \Bigr] \, ,
\label{eq:local-ergo}
\end{eqnarray}
where $\overline{U}_{a}$ is the optimal unitary operation performed by $A$, $d_a$ is dimension of the subsystem of $S$ associated to $A$, and $\mathbb{1}_b$ is an identity operator on $\mathcal{H}_b$. Therefore, for a multipartite system $S$ composed of $n$ subsystems shared among $\{A_i\}_{i=1}^n$ agents, we formalize the maximal amount of work that can be extracted from $S$ by $A_i$ as \textit{local ergotropy functional on $A_i$}
\begin{eqnarray}
\nonumber \mathcal{E}_{\rm La_i}(\rho, H) &:=& \max_{U_{a_i} \in \mathtt{U}(d_{a_i})}\Bigl( \tr\Bigl[ ( \rho - U_{\rm La_i} \rho U_{\rm La_i}^\dag ) H \Bigr] \Bigr) \\
&=& \tr\Bigl[ ( \rho - \overline{U}_{\rm La_i}\rho \overline{U}_{\rm La_i} ) H \Bigr] \, ,
\label{eq:local-ergo-N}
\end{eqnarray}
where we denote $U_{\rm La_i} := \mathbb{1}_{a_1} \otimes \ldots \otimes U_{a_i} \otimes \ldots \otimes \mathbb{1}_{a_n}$ and $\overline{U}_{\rm La_i} := \mathbb{1}_{a_1} \otimes \ldots \otimes \overline{U}_{a_i} \otimes \ldots \otimes \mathbb{1}_{a_n}$, with $\overline{U}_{a_i}$ being the corresponding optimal unitary operation on $A_i$.

\section{Parallel ergotropy}
\label{sec:parallel-def}
\begin{figure}
    \centering
    \begin{overpic}[width=.45\textwidth]{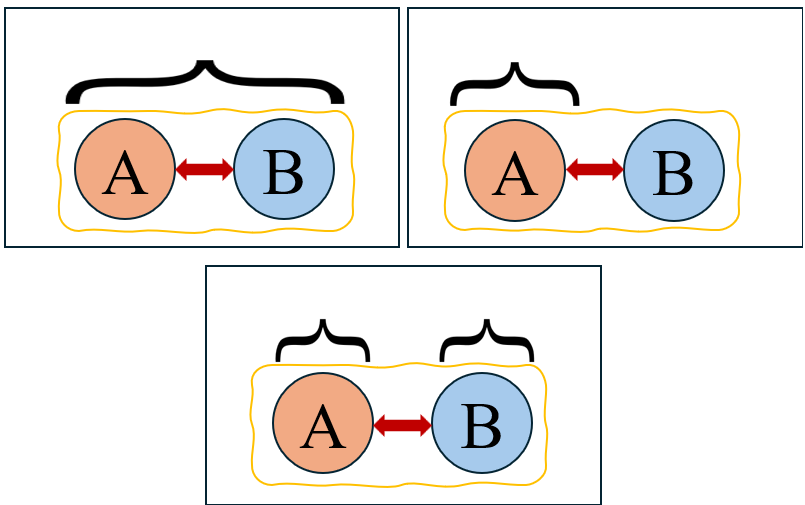}
    \put(8,57){{(global) ergotropy}}
    \put(55,57){{local ergotropy}}
    \put(34,26){\textit{parallel ergotropy}}
\end{overpic}
    \caption{Schematic representation for a bipartite system of \textit{parallel ergotropy}: concurrent local unitaries on A and B; compared to (global) ergotropy \cite{allahverdyan2004maximal}: a single global unitary; and local ergotropy from Ref.\,\cite{LErgo}: local unitary on A only. The parts A and B interact via a Hamiltonian term represented by a red double arrow. Unitaries are depicted as curly brackets and the joint AB state by the yellow envelope.}
    \label{fig:schema}
\end{figure}
Let us consider a multipartite quantum battery consisting of many interacting subsystems $S_1, \ldots, S_N$ shared between $N$ agents. As anticipated, performing generic global unitaries on the whole system might be experimentally challenging, and it is reasonable to question work extraction under restriction to operations performed locally by each agent.
While local ergotropy \eqref{eq:local-ergo} quantifies work that can be extracted by a single agent acting on its subsystem, it cannot catch cooperative work extracting strategies of the agents. In order to fill this gap, we introduce the \textit{parallel ergotropy} (PE) as the maximum energy extractable via concurrent local unitary operations,
\begin{equation}\label{eq:parallel-ergo-N}
    \mathcal{E}_{\rm P}(\rho,H) := \max_{U_{\mathrm{P}} \in \mathcal{U}_{\mathrm{P}}^N} \tr  \Bigl[(\rho - U_{\mathrm{P}} \rho U_{\mathrm{P}}^\dag) H\Bigr],
\end{equation}
where $\mathcal{U}_{\rm P}^N$ denotes the set of {\it $N$-parallel unitaries}, i.e., unitaries of form 
\begin{equation}
\label{setUp}
    U_{\mathrm{P}} := \otimes_{i=1}^N U_i\,,
\end{equation}
with $U_i \in \mathtt{U}(d_i)$ acting on $i$-th subsystem.
{A schematics of \textit{parallel ergotropy} for a bipartite system, compared to (global) ergotropy \cite{allahverdyan2004maximal} and local ergotropy from Ref.\,\cite{LErgo}, 
is reported in Fig.\,\ref{fig:schema}.} 

A comment on the definition \eqref{eq:parallel-ergo-N} is necessary. As for the local ergotropy \eqref{eq:local-ergo}, the realization of such local unitary evolutions is intended on time scales that are much smaller than any other time scale in the problem.
Put differently,
functional \eqref{eq:parallel-ergo-N} does not allow for slow driving of the local Hamiltonian and thus is time-dependent on the free dynamics of the system. In formulas, defining $U_0(t):=\exp (-i H t) $,  
\begin{equation}
\mathcal{E}_{\rm P}({U_{0}(t)~ \rho~ U_{0}(t)^\dag},H) \neq \mathrm{const}\,.
\end{equation} 
Therefore, PE can be regarded as the figure of merit for work extraction only if operations are carried out instantaneously with respect to the dynamics of the system.

We point out two main reasons for which one may stick to this restricted class of operations instead of the bigger set $\mathcal{U}_{\mathrm{Pex}} $ introduced in Section \ref{EPE}, which exploits both \textit{parallel} (hence, local on each subsystem) Hamiltonian control and internal evolution induced $H$. First, one may need to charge and discharge the battery on a faster time-scale than the one characterizing the dynamics of the system.
Second, the precise time-keeping necessary to perform \textit{free evolution} operations has its own energy and thermodynamic cost \cite{CloksCost, Q.Cloks}, which might be higher than the advantage of reaching generic global unitaries over the local ones. Notice that both arguments are also valid for the local ergotropy introduced in Ref.~\cite{LErgo}.

In the rest of the paper, we mainly consider a bipartite quantum battery (i.e., $N=2$). Denoting both parties $A$ and $B$ with associated Hilbert spaces $\mathcal{H}_a$ and $\mathcal{H}_b$, respectively. It is described hence by a bipartite state $\rho_{ab} \in D(\mathcal{H}_a \otimes \mathcal{H}_b)$ and equipped with a Hamiltonian
\begin{equation}
\label{hamilt}
    H = H_{a} + H_{b} + V_{ab} := H_{ab}\,,
\end{equation}
where $V_{ab}$ is the interaction term acting on the joint Hilbert space $\mathcal{H}_{a} \otimes \mathcal{H}_{b}$, and $H_{a}$ and $H_{b}$ are the local terms on $A$ and $B$, respectively. Without loss of generality, we will choose the tree quantities such that $ \tr_{\rm a}[V_{ab}] = \tr_{\rm b}[V_{ab}] = 0$. In what follows, we study the properties of the considered quantum battery with respect to action of 2-parallel unitaries $U_a \otimes U_b$, where $U_a \in \mathtt{U}(d_a)$ and $U_b \in \mathtt{U}(d_b)$. In particular, the PE functional \eqref{eq:parallel-ergo-N} reduces to
\begin{equation}
\label{eq:parallel-ergo}
\mathcal{E}_{\rm P}(\rho_{ab}, H_{ab}) = \max_{U_a, U_b}\tr  [ (\rho_{ab} - (U_a \otimes U_b) \rho_{ab} (U_a^\dag \otimes U_b^\dag))H_{ab}]\,.
\end{equation}
Our choice is due to two main reasons. On the one hand, it allows for a more straightforward comparison with the local ergotropy functional \eqref{eq:local-ergo}. On the other hand, despite its higher experimental accessibility compared with the global ergotropy \eqref{eq:ergo}, both the analytical and numerical treatments of PE appear to be already rather involved for $N=2$.

\section{General parallel bounds and egoistic work extracting strategies}
\label{sec:bounds}
In order to simplify expressions involving several unitary operations, given a state $\rho$, we sometimes represent action of a unitary operator $U$ as a superoperator $\mathsf{U}[\rho] := U\rho U^{\dagger}$. By construction, the following hierarchy of the ergotropic quantities introduced above can be established,
\begin{widetext}
\begin{equation}\label{CooperationGap}
\mathcal{E}(\rho, H) \geq \mathcal{E}_{\rm P}(\rho, H) \geq \mathcal{E}_{\rm La}(\rho, H) + \mathcal{E}_{\rm Lb}(\overline{\mathsf{U}}_a[\rho], H) \geq \mathcal{E}_{\rm La}(\rho, H)\,,
\end{equation}
\end{widetext}
where $\overline{\mathsf{U}}_a $ corresponds to the optimal unitary $\overline{U}_a$ on $A$ resulting from maximization in \eqref{eq:local-ergo}. Notice that the hierarchy \eqref{CooperationGap} is constructed with respect to the party $A$. Due to the symmetry, it remains naturally valid with respect to $B$ as well, i.e., when $a$ and $b$ are exchanged in the ergotropic quantities of \eqref{CooperationGap}.

The new quantity $\mathcal{E}_{\rm La}(\rho, H)+\mathcal{E}_{\rm Lb}(\overline{\mathsf{U}}_a[\rho], H) $ in \eqref{CooperationGap} represents the maximum work that the agents $A$ and $B$ can extract from the system under an \textit{egoistic} strategy. This kind of strategy implies a fixed order in which agents act on the bipartite system, so that $A$ applies the optimal unitary $\overline{U}_a$ maximizing local work extraction with respect to its subsystem, and then $B$ extracts work from the resulting state. This contrasts with the PE implying a \textit{cooperative} strategy, under which $A$ and $B$ agree on applied local unitaries, optimizing thereby total extracted work. This leads to a natural question: can $A$ and $B$ follow egoistic strategies in order to extract work guaranteed by PE? Generally speaking, the answer is no despite a misleading intuition suggested by commutativity of local unitaries of $A$ and $B$.

Let us assume that the competing agents $A$ and $B$ share a 2-qubit battery equipped with a Hamiltonian \eqref{hamilt} and can either cooperate (achieving the value of PE) or not on choosing the local work extracting unitary operations. We proceed to show that, for most $H_{ab}$, there exist states, for which a cooperative strategy is strictly necessary to achieve the maximal work $\mathcal{E}_{\rm P}(\rho_{ab}, H_{ab})$, and the following strict inequalities hold,
\begin{eqnarray}
\label{SwapCooperationGapA}
\mathcal{E}_{\rm P}(\rho_{ab}, H_{ab}) \!\! &>& \!\! \mathcal{E}_{\rm Lb}(\rho_{ab}, H_{ab})+\mathcal{E}_{\rm La}(\overline{\mathsf{U}}_b[\rho_{ab}], H_{ab})\,,\quad \\
\mathcal{E}_{\rm P}(\rho_{ab}, H_{ab}) \!\! &>& \!\! \mathcal{E}_{\rm La}(\rho_{ab}, H_{ab})+\mathcal{E}_{\rm Lb}(\overline{\mathsf{U}}_a[\rho_{ab}], H_{ab}) \,.\,\, \label{SwapCooperationGapB}
\end{eqnarray}

An interesting feature is that, sometimes, the first agent acting on the system has to inject work in order to enable the other agent to complete the optimal cooperative gate. To show this, let us we pick here a specific pair of states and Hamiltonians (see Appendix \ref{app:egoistic} for analysis of egoistic strategies for a larger class of states and Hamiltonians).
Let us assume that the quantum battery is equipped with the Hamiltonian $H_{ab}$ characterized by the following eigenvectors,
\begin{eqnarray}
&&\ket{E_{0}}:=\ket{00}  ,\, 
 \ket{E_{1}} :=\ket{11},\, \\
&&\ket{E_{2}}:=\ket{10},\, 
 \ket{E_{3}}:=\ket{01} ,\,      
\end{eqnarray}
where $E_{j}$ denote the corresponding energies fulfilling $E_{0} < E_{1} < E_{2} \leq E_{3}$, and $\{|0/1\rangle\}$ is the computational basis of each qubit. In turn, the state of the battery is given by $ \rho_{ ab} = \ket{E_{1}}\bra{E_{1}}$. In this case, the egoistic strategies are highly inefficient: it is straightforward to verify that any single local unitary increases the energy of the battery. Therefore, the optimal egoistic strategy in this setup is not implementing any work extracting unitary by each agent at all, and 
\begin{equation}
    \mathcal{E}_{\rm La}(\rho_{ab}, H_{ab}) = \mathcal{E}_{\rm Lb}(\rho_{ab}, H_{ab}) = 0.    
\end{equation}
On the other hand, if $A$ and $B$ are allowed to cooperate, a bit flip on each subsystem $X \otimes X$ can be performed, reaching thereby the ground state and extracting PE (which is equal to global ergotropy \eqref{eq:ergo} as well)
\begin{equation}
    \mathcal{E}_{\rm P}(\rho_{ab}, H_{ab}) = \mathcal{E}(\rho_{ab}, H_{ab}) = E_{1}-E_{0}>0. 
\end{equation}
Apart from providing an example of a setup that opens a gap between egoistic and cooperative strategies via \eqref{SwapCooperationGapA}--\eqref{SwapCooperationGapB}, this result demonstrates several interesting features. First of all, we notice that it does not involve coherence, being hence a \textit{classical} result.
On the other hand, it suggests that, depending on the order in which the gates are applied, an optimal strategy to extracted the maximal total work requires that one agent has to invest work, while another receives the entire work gain. Indeed, let us assume that $A$ acts first on its subsystem. Then it can apply a bit flip $X$, reaching thus the state $|E_3\rangle$ and paying an energy cost $E_3-E_1$. After that, $B$ applies a bit flip $X$ on its subsystem, reaching the ground state and receiving the gain $E_3 - E_0$. In turn, the total gain is $E_1-E_0$ as anticipated. Generally speaking, this results form the fact that
\begin{equation}
\nonumber E(\rho_{ab})-E(\mathsf{U}_{a}[\rho_{ab}]) \neq  E(\mathsf{U}_{b}[\rho_{ab}]) - E((\mathsf{U}_{a}\otimes \mathsf{U}_{b})[\rho_{ab}])\,,
\end{equation}
where $E(\rho):={\rm tr} [H \rho]$.

At first, this might seem in contrast with the non-signaling principle. Let us assume that $A$ and $B$ find themselves in laboratories separated by a distance $L$. If $A$ performs $U_{a}$ just after Bob has performed $U_{b}$, amount of extracted work by $A$ could indicate the choice of $U_{b}$. However, it is necessary to take into account that the non-local Hamiltonian $H_{ab}$ is a valid approximation of the real interaction only on time scales longer than $L/c$. Therefore, the effective non-relativistic $H_{ab}$ does not provide a valid description of the dynamics of the system, so that the energy of the quantum battery is not given by $\operatorname{tr}[H_{ab}\rho]$ if $A$ and $B$ perform the local work extracting unitaries in space-separated events. On time scales longer than $\frac{L}{c}$ the previous example is instead valid, and energy can be transferred among the agents. Indeed, LOCC protocols for energy transmission are known as quantum energy teleportation and extensively studied, e.g., in \cite{ETele}.

\section{Parallel bounds in specific cases} 
\label{sec:LMM}
Despite its operational simplicity, generally speaking, PE is a complicated functional to compute. 
In what follows, we provide analytic formulae and bounds for PE for a relevant class of states and Hamiltonians which are computable. 

\subsection{Locally maximally mixed states and/or Hamiltonians with null local terms}
First, we focus on the class of quantum batteries, which can be described (i) by \textit{locally maximally mixed} (LMM) states, that play an important role in quantum information, and/or (ii) in the limit of asymptotically strong Hamiltonian interaction with respect to the Hamiltonian local terms. 

In the case (i), we consider a bipartite LMM state $\rho_{ab} \in D(\mathcal{H}_a \otimes \mathcal{H}_b)$, where $\operatorname{dim}\mathcal{H}_a = \operatorname{dim}\mathcal{H}_b = d $.
This means that the reduced density operators correspond to maximally mixed states \cite{LMM},
\begin{align}\label{eq:LMM}
&\tr_a[\rho_{ab}] = \frac{\mathbb{1}_b}{d}\,, \quad
\tr_b[\rho_{ab}] = \frac{\mathbb{1}_a}{d}\,.
\end{align}
The set of LMM states contains locally maximally entangled states \cite{Bryan2018, Bryan2019}, which are pure states $\rho_{ab}$ satisfying \eqref{eq:LMM}, and form a convex set including pure maximally entangled states as well as mixed states such as Werner states. 
Qubit Bell states are known to be the only extremal points of the set of LMM states for $d=2$. 
For higher dimensions, the set of LMM states is strictly larger: e.g., for $d = 3$ its extremal points are not necessarily locally equivalent \cite{LMM}.

{As detailed in Appendix \ref{ProofLLMBound}, the analysis carried below holds in the case (ii) as well, corresponding to the condition 
\begin{equation}
        H_a=H_b=0
        \label{null-loc-ham}
\end{equation}
on the Hamiltonian. 
In our notation \eqref{hamilt}, it can be seen as the limit of asymptotically strong interaction, in which local contributions are neglected with respect to the interaction term.
A paradigmatic example of a Hamiltonian that satisfies these conditions is provided by the antiferromagnetic Hamiltonian 
\begin{equation}\label{eq:AntiferrH}
    H_{\mathrm{Ant}} := \omega (\Vec{\sigma}_{a} \cdot \Vec{\sigma}_{b}),
\end{equation}
where $\Vec{\sigma}_{a/b}$ are vectors, whose components $\sigma_{a/b}^j$ are generalized Pauli operators \cite{GOP} on $\mathcal{H}_{a/b}$, which reduce to the usual Pauli operators for $d=2$.
}

In order to analyze PE for a quantum battery satisfying (i) and/or (ii), we do the following observation. Exploiting the generalized Pauli operator expansion \cite{GOP}, the action of work extracting unitaries $U_{\rm P}$ in \eqref{eq:parallel-ergo} can be associated to rotations in real space (see Appendix \ref{GPO}). 
The latter form a subgroup $\mathcal{SO}(d^2-1)$ of the corresponding rotation group $\mathtt{SO}(d^2 - 1)$, i.e. $\mathcal{SO}(d^2-1) \subseteq \mathtt{SO}(d^2 - 1)$. 
Therefore, the optimization over parallel unitaries in \eqref{eq:parallel-ergo} can be substituted by optimization over $\mathcal{SO}(d^2-1)$, and PE in the case of systems satisfying (i) and/or (ii) is given by
\begin{equation} \label{eq:optOTV}
    \mathcal{E}_{\rm P}(\rho_{ab}, H_{ab}) = \operatorname{tr}[\mathbf{V}\mathbf{T}] + \max_{\mathbf{O}_a,\;\mathbf{O}_b} \operatorname{tr}[(-\mathbf{V})\mathbf{O}_a\mathbf{T}\mathbf{O}_{b}^{T}],
\end{equation}
where $\mathbf{T}$ and $\mathbf{V}$ are two-body components of the state $\rho_{ab}$ and the Hamiltonian $H_{ab}$, respectively, i.e.,
\begin{align}
\mathbf{T}_{i,j} &:=  \tr[\rho_{ab}(\sigma_{a}^i \otimes \sigma_{b}^j)]. \\
\mathbf{V}_{i,j} &:=  \frac{1}{4}\tr[V_{ab}(\sigma_{a}^j \otimes \sigma_{b}^i)], 
\end{align}
and $\mathbf{O}_{a,b} \in \mathcal{SO}(d^2-1)$.
The optimization in \eqref{eq:optOTV} can be carried out analytically (details are reported in Appendix \ref{ProofLLMBound}) over the entire rotation group $\mathtt{SO}(d^2 - 1)$, establishing the following upper bounds on PE,
\begin{widetext}
\begin{equation}
\label{eq:LMMBound}
\mathcal{E}_{\rm P}(\rho_{ab}, H_{ab}) \leq 
\begin{cases}
\tr[\rho_{ab} V_{ab}] + (\vec \lambda^{|\mathbf{V}|}\cdot \vec \lambda^{|\mathbf{T}|}) & \text{if $\det(-\mathbf{V}\mathbf{T}) \geq 0$}, \\
\tr[\rho_{ab} V_{ab} ] + (\vec \lambda^{|\mathbf{V}|}\cdot \vec \lambda^{|\mathbf{T}|}-2\lambda_0^{|\mathbf{V}|}\lambda_0^{|\mathbf{T}|}) & \text{if $\det(-\mathbf{V}\mathbf{T}) < 0$},
\end{cases}
\end{equation}
\end{widetext}
where $\vec{\lambda}^{|\mathbf{M}|} = \{ \lambda_0^{|\mathbf{M}|} \leq \ldots \leq \lambda_{d^2-2}^{|\mathbf{M}|} \}$ denotes a non-decreasingly ordered set of singular values of a generic matrix $\mathbf{M}$, i.e., eigenvalues of $|\mathbf{M}| = \operatorname{tr}[\sqrt{\mathbf{M}^\dagger \mathbf{M}}]$. Since there is a one-to-one correspondence (up to a global phase) between qubit unitary transformations and 3D rotations due to the isomorphism $\mathtt{U}(2)/\mathtt{U}(1) \cong \mathtt{SO}(3)$, the upper bound \eqref{eq:LMMBound} is saturated in the case of a two-qubit system.
\subsection{Werner states}
\label{sec:werner}
Werner states $W$ are bipartite LMM states that are invariant under parallel unitary operations of form $U \otimes U$,
\begin{equation}
\label{W-invariance}
    (\mathsf{U} \otimes \mathsf{U})[W] = W,
\end{equation}
for any $U$ acting on the $d$-dimensional Hilbert space. Werner states provide an important example of states whose PE and local ergotropies coincide, i.e., for a bipartite system $AB$,
\begin{equation}
\label{WPEeqLE}
    \mathcal{E}_{\rm P}(W, H_{ab}) = \mathcal{E}_{\rm La}(W, H_{ab}) = \mathcal{E}_{\rm Lb}(W, H_{ab}).
\end{equation}
In order to see this, let us assume $\overline{U} \otimes \overline{V}$ to be an optimal parallel unitary that achieves the value of PE $\mathcal{E}_{\rm P}(W)$. Taking into account the invariance condition \eqref{W-invariance} with $U = \overline{V}$ and $U = \overline{U}$, we obtain
\begin{eqnarray}
    \nonumber (\overline{\mathsf{U}} \otimes \overline{\mathsf{V}})[W] &=& (\overline{\mathsf{U}}\overline{\mathsf{V}}^\dagger \otimes \mathbb{1})\Bigl[(\overline{\mathsf{V}} \otimes \overline{\mathsf{V}})[W]\Bigr] \\
    &=& (\overline{\mathsf{U}}\overline{\mathsf{V}}^\dagger \otimes \mathbb{1}) [W],
\end{eqnarray}
and
\begin{eqnarray}
    \nonumber (\overline{\mathsf{U}} \otimes \overline{\mathsf{V}})[W] &=& (\mathbb{1} \otimes \overline{\mathsf{V}}\overline{\mathsf{U}}^\dagger)\Bigl[(\overline{\mathsf{U}} \otimes \overline{\mathsf{U}})[W]\Bigr] \\
    &=& (\mathbb{1} \otimes \overline{\mathsf{V}}\overline{\mathsf{U}}^\dagger)[W],
\end{eqnarray}
so that the same amount of work can be extracted via a local unitary $\overline{U}^\dag \overline{V}$ on the single subsystem $A$ as well as $\overline{V}^\dag \overline{U}$ on the subsystem $B$, hence, proving the property \eqref{WPEeqLE}. 

As an example, let us consider a Werner state of two qubits, which can be written explicitly as
\begin{equation}\label{eq:WernerState}
      W_p := (1-p)\frac{\mathbb{1}_{ab}}{4}+p\ket{\psi_{max}}\bra{\psi_{max}}\,,
\end{equation}
with $p \in [0,1]$ and $\ket{ \psi_{max}}$ being any of the four Bell states $|\Phi^\pm\rangle = \frac{1}{2}(|00\rangle \pm |11\rangle)$ and $|\Psi^\pm\rangle = \frac{1}{2}(|01\rangle \pm |10\rangle)$. 
The state in \eqref{eq:WernerState} is characterized by $\mathbf{T} = - pO$, where $O \in \mathtt{SO}(3)$ \cite{2-Qubits}.
Therefore, applying the upper bound \eqref{eq:LMMBound}, which is saturated since a 2-qubit state is considered, we obtain PE (as well as local ergotropy \cite{LErgo}) of the Werner state \eqref{eq:WernerState},
\begin{widetext}
\begin{equation}
\label{WernerParallel}
\mathcal{E}_{\rm P}(W_p, H_{ab}) = 
\begin{cases}
p \big ( \bra{\psi_{max}}V_{ab}\ket{\psi_{max}} + \operatorname{tr}[|\mathbf{V}|] \big ) & \text{if $\det(\mathbf{V}) \geq 0$}, \\
p \big( \bra{\psi_{max}}V_{ab}\ket{\psi_{max}} + (\operatorname{tr}[|\mathbf{V}|] - 2\lambda_0^{|\mathbf{V}|}) \big ) & \text{if $\det(\mathbf{V}) < 0$}.
\end{cases}
\end{equation}    
\end{widetext}
In particular, for the antiferromagnetic Hamiltonian \eqref{eq:AntiferrH}, which is characterized by $\mathbf{V} = \omega \mathbb{1}$, it reads:
\begin{equation}
    \mathcal{E}_{\rm P}(W_p, H_{\mathrm{Ant}}) = 
\begin{cases}
4p\omega & \text{if $|\psi_{max}\rangle \in \{|\Phi^\pm\rangle, |\Psi^+\rangle\}$}, \\
0 & \text{if $|\psi_{max}\rangle =|\Psi^-\rangle$}.
\end{cases}
\end{equation}

\section{Numerical Approximation via Semidefinite Programming}
\label{sec:SDP}
\subsection{SDP method}
As specified in the previous section, the analytical upper bound on PE provided by Eq.\,\eqref{eq:LMMBound} is only valid when either
(i) the Hamiltonian has null local terms [Eq.\,\eqref{null-loc-ham}] and/or
(ii) the state is locally maximally mixed [Eq.\,\eqref{eq:LMM}].
We derive now a numerical approach allowing one to obtain another upper bound on PE, which is general, i.e. not restricted to requirements (i) and/or (ii).

Inspired by Ref.\,\cite{LErgo}, we first rely on semi-definite programming, the Choi-Jamio\l kowski isomorphism, and 
the relaxation from unitary to unital channels. 

Semi-definite programming is a subclass of convex optimization problems.
The quality of these approaches is two-fold.
First, 
local optima are also global optima for convex optimization problems.
Second, 
semi-definite programs are characterized by constraints in the form of linear matrix inequalities and for which efficient algorithms exist for its numerical solution \cite{Boyd_Vandenberghe_2004, Cavalcanti_Skrzypczyk_SDP}, with polynomial complexity in the number of constraints and on the size of the matrix.
However, in general, 
rephrasing the PE optimization problem
\eqref{eq:parallel-ergo-N} in terms of a semi-definite program (SDP), 
comes at the cost of enlarging the set of allowed operators, hence leading to an upper bound for PE.
Indeed, the original problem \eqref{eq:parallel-ergo-N} 
is not an SDP, because of the nonlinear constraints
[unitarity and tensor product structure in \eqref{setUp}]
concerning the quantum channel that should be applied on the system.
%

To construct a suitable SDP we exploit the 
Choi-Jamio\l kowski isomorphism \cite{Jamiolkowski_1972}.
The \textit{Choi state} $J_\Lambda$ of the channel $\Lambda$ is obtained as
\begin{equation}
J_\Lambda = \sum_{\mu,\nu} \,|\mu\rangle\langle\nu| \otimes \Lambda(|\mu\rangle\langle \nu|)\,, 
\end{equation}
and the action of the channel $\Lambda$ from
\begin{equation}
\Lambda(\rho) = \mathrm{Tr}_{\mathtt{in}}\left[ J_\Lambda\,\left(\rho^T \otimes \mathbb{1} \right) \right]\,,
\end{equation}
expressing the one-to-one correspondence between $\Lambda$ and $J_\Lambda$.
More specifically,
defining the subsystems' Hilbert spaces as 
$\mathcal{H}^{\mathtt{out}}_i$
and 
ancillary copies of them as $\mathcal{H}^{\mathtt{in}}_i$, with $i\in\{1,\dots,n\}$, 
$\Lambda$ is a channel acting from $\mathcal{L}(\otimes_{i=1}^n\,\mathcal{H}^{\mathtt{in}}_i)$ to $\mathcal{L}(\otimes_{i=1}^n\,\mathcal{H}^{\mathtt{out}}_i)$, and $J_\Lambda$ is an operator in the combined space $\mathcal{L}\left[\left(\otimes_{i=1}^n \mathcal{H}^{\mathtt{in}}_i\right) \otimes \left(\otimes_{i=1}^n \mathcal{H}^{\mathtt{out}}_i\right)\right]$.

The SDP we introduce reads as,
\begin{align}
    J^{\rm( opt)}_\Lambda=\argmin_{J_\Lambda} &\quad \mathrm{Tr}\left[J_\Lambda\, \left(\rho^T \otimes H\right)\right] \nonumber \\
    \textrm{s.t.}&\quad J_\Lambda \geq 0,\nonumber \\
    &\quad \mathrm{Tr}_{\mathtt{out}}[J_\Lambda] = \mathbb{1}_{\mathtt{in}}, \nonumber \\
    \forall i \in (1,\ldots,n) &\quad \mathrm{Tr}_{\mathtt{in}^{(i)}}[J_\Lambda] = \mathrm{Tr}_{\mathtt{in}^{(i)}, \mathtt{out}^{(i)}}[J_\Lambda] \otimes \frac{\mathbb{1}_{\mathtt{out}^{(i)}}}{d_i}, \nonumber \\
    &\quad J_\Lambda\,\in\,\mathsf{DPS}_k
    \label{eq:SDP_upper_bound}
\end{align}
and the corresponding upper bound for PE is given by
the rhs of the following inequality:
\begin{equation}
\label{sdp-bound-explicit}
    \mathcal{E}_{\rm P}(\rho_{ab}, H_{ab}) \leq 
\mathrm{Tr}[\rho\,H] - \mathrm{Tr}[J^{\rm( opt)}_\Lambda\, \left(\rho^T \otimes H\right)]\,.
\end{equation}
We describe now the main points of the method, leaving further details to Appendix \ref{app:approx_constraints}.

Positive semi-definiteness condition [first constraint in\,\eqref{eq:SDP_upper_bound}] and trace-preserving condition [second constraint in\,\eqref{eq:SDP_upper_bound}] ensure that the super-operator $\Lambda$ associated to $J_\lambda$ is a Completely Positive Trace-Preserving map. 

The remaining constraints 
are relaxations of the nonlinear constraints present in the original problem \eqref{eq:parallel-ergo-N}.
To start, for any given set of quantum channels $\mathcal{C}$ let us denote with 
\begin{eqnarray}
\label{conv-def}
    \textit{Conv}(\mathcal{C}):= \left \{ \underset{i}{\sum}  p_{i}\mathbf{\Phi}_{i}:  \underset{i}{\sum}  p_{i}=1 , p_{i}\geq 0 , \mathbf{\Phi}_{i}\in \mathcal{C} \right \} 
    \end{eqnarray}
    the convexification of $\mathcal{C}$.
Analogously to Ref.\,\cite{LErgo}, we first relax local unitarity with local unitality, that for the Choi state means imposing the constraints present in the second last line of \eqref{eq:SDP_upper_bound} (see Appendix \ref{app:approx_constraints} for details).
Actually, 
we have 
\begin{equation}
\label{inclusions-sets}
   \otimes_{i=1}^{n}\mathcal{U}_i \subseteq  
   \otimes_{i=1}^{n}Conv(\mathcal{U}_i)
   \subseteq
   \otimes_{i=1}^{n}\mathcal{U}^{\rm ni}_i
      \subseteq
Conv(\otimes_{i=1}^{n}\mathcal{U}^{\rm ni}_i)\,,
\end{equation}
i.e.,
enlarging the set at each step, we pass from the set $\otimes_{i=1}^{n}\mathcal{U}_i$ of tensor products of unitaries, 
to the set $\otimes_{i=1}^{n}Conv(\mathcal{U}_i)$ of tensor products of convex mixtures of unitaries, 
then to the set $\otimes_{i=1}^{n}\mathcal{U}^{\rm ni}_i$ of tensor products of unital maps,
and finally to the set $Conv(\otimes_{i=1}^{n}\mathcal{U}^{\rm ni}_i)$
of convex mixtures
of 
tensor products of unital maps.
Local unitality means that the channel is independently unital on each subsystem it acts upon, 
\begin{equation}
    \Lambda(\rho \otimes \mathbb{1}_i) = \rho' \otimes \mathbb{1}_i,
    \label{eq:def_local_unital}
\end{equation}
where $\rho$ is a density matrix in $\mathcal{B}(\otimes_{j \neq i} \mathcal{H}^j_{\mathtt{in}})$ and $\rho'$ the resulting state also acting on all but the i-th output Hilbert space.
Satisfying \eqref{eq:def_local_unital} for any $i\in\{1,2,\dots,n\}$ means $\Lambda \in Conv(\otimes_{i=1}^{n}\mathcal{U}^{\rm ni}_i)$ and is encoded in the constraints present in the second-to-last line of \eqref{eq:SDP_upper_bound} (see Lemma\,\ref{lemma-loc-unital} for a detailed proof).

It should be noted that both passing from 
$\otimes_{i=1}^{n}\mathcal{U}_i$ to $  
\otimes_{i=1}^{n}Conv(\mathcal{U}_i)$
(see Lemma\,\ref{lemma-convconv})
and passing 
from 
$\otimes_{i=1}^{n}\mathcal{U}^{\rm ni}_i$ to  
$Conv(\otimes_{i=1}^{n}\mathcal{U}^{\rm ni}_i)$
doesn't imply an approximation, as the optimum value does not change when enlarging the set in these ways.
Indeed,
the objective function
is linear in the Choi state and, therefore, we can retain as extreme points products of Choi states
(see Appendix \ref{clarifications} for extended proofs).
Instead, passing from
$\otimes_{i=1}^{n}Conv(\mathcal{U}_i)$
to 
$\otimes_{i=1}^{n}\mathcal{U}^{\rm ni}_i$
implies an approximation in the general case.
Interestingly, in the particular case where all the subsystems are qubits, the latter is not an approximation either \cite{Mendl_2009}.

Finally, the constraints concerning the last line of \eqref{eq:SDP_upper_bound} correspond to the $k$-th level of the Doherty-Parrilo-Spedalieri (DPS) hierarchy for separable states \cite{DPS_2004}, based on symmetric extensions of the state with positive partial transpose (PPT).
These constraints are aimed at imposing the tensor product structure of the channel $\Lambda=\otimes_{i=1}^n {\Lambda_i}$, implying the Choi state 
is a tensor product of Choi states:
$J_{\Lambda}=\otimes_{i=1}^n J_{\Lambda_i}$,
each $J_{\Lambda_i}$ acting on subsystem $\mathcal{H}^{\mathtt{in}}_i \otimes \mathcal{H}^{\mathtt{out}}_i$.
%
Since ensuring such a structure is also usually a difficult task computationally \cite{EntanglementNPHard}, a double relaxation is employed on $J_\Lambda$. 
Similarly to before, optimizing over convex mixtures of product Choi states is equivalent and means imposing separability condition between all the pairs of different subsystems. 
Then the method of Ref.\,\cite{DPS_2004} allows one to approximate from the outside the target convex set with the converging hierarchy of simpler tests. 
In particular, although higher levels of the DPS hierarchy are computationally too resource demanding, already the \enquote{level zero} of just requiring PPT for $J_\Lambda$ results usually in a good approximation to the target set in practical cases.
{Since the hierarchy of approximations converges to the actual set of separable states, the higher the level of the hierarchy implemented, the closer should be the map to an actual combination of products of unital maps.

In general, the whole scheme provides an upper bound for PE that is not guaranteed to converge to the exact value even for high level $k$.
Such bound can happen to be even looser than the one coming from approximating parallel ergotropy with global ergotropy.
In this sense, a practical way to control the accuracy of the result is to monitor the purity of the Choi state resulting from the optimization \eqref{eq:SDP_upper_bound}: the greater the purity, the smaller the price we have paid by relaxing convex combinations of local unitaries with local unitals [intermediate step in \eqref{inclusions-sets}].
Notably, for increasing DPS level $k$, convergence to the actual value of PE is instead ensured 
when the subsystems are single qubits.
Indeed, as in the latter case the equivalence 
$\otimes_{i=1}^{n}Conv(\mathcal{U}_i)
   =
   \otimes_{i=1}^{n}\mathcal{U}^{\rm ni}_i$
holds true,
the SDP method is asymptotically exact.
 
However, for brevity, in what follows we shall only test
the SDP method
for $n=2$ subsystems and
 level $k=0$ of the DPS hierarchy, i.e. using 
\begin{equation}
\label{dps0}
J_\Lambda \in {\rm DPS}_0
\iff
    J^{T_i}_{\Lambda} \geq 0\quad \forall i \in \{1,2,\dots,n\}\,,
\end{equation}
 meaning that the partial transposition on the space of each subsystem $(\mathtt{in}^{(i)}, \mathtt{out}^{(i)})$ should preserve positive semi-definiteness.

\subsection{Applications}
\label{sec:sdp-appl}

We present now some test cases to compare the different methods for estimating PE. For direct comparison, we restrict the analysis to cases where the analytical upper bound \eqref{eq:LMMBound} is applicable, first on the Werner state, then on generic states but using Hamiltonians with no local terms. 

Results for the first case are presented in Fig.\,\ref{fig:algo_comparison_Werner}, where a range of upper bounds for PE--corresponding to global ergotropy, the SDP approximation from Eq.\,\eqref{eq:SDP_upper_bound} under
$0$-th level of the DPS hierarchy \eqref{dps0},
and the analytical bound of \eqref{eq:LMMBound}--is calculated for different visibilities $p$ of the two-qubit Werner state [Eq. \eqref{eq:WernerState}], under the action of the Hamiltonian
\begin{equation}
\label{flip-flop}
 H = \frac{1}{2}\omega_a\,\sigma^z_a + \frac{1}{2}\omega_b\,\sigma^z_b + \frac{1}{2}g\,(\sigma^x_a \otimes \sigma^x_b + \sigma^y_a \otimes \sigma^y_b)\,.
\end{equation}
Also a lower bound is computed through direct optimization of local unitaries (see Appendix\,\ref{app:DirectUnitary} for details). 
Coincidence of the latter with the bound of \eqref{eq:LMMBound} confirms that the estimate is exact, as expected.
Instead, the SDP method shows at least an improvement in comparison to upper bounding PE using global ergotropy,
although not closing the gap with the actual PE value. 

\begin{figure}[t]
\includegraphics[width=0.8\linewidth]{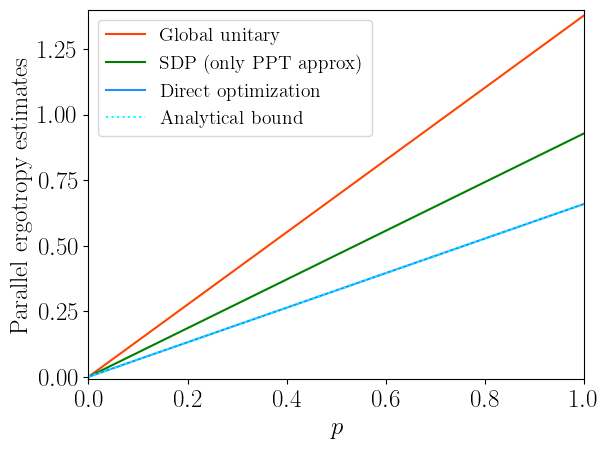}
    \caption{Comparison among different estimates of PE for the two-qubit Werner state and the Hamiltonian \eqref{flip-flop}.
    We set the Hamiltonian parameters as
    $\omega_a = 1.0,\,\omega_b = 1.1,\,g = 0.33$
    and plot the results as function of the visibility $p$ of the Werner state.
The analytical bound of Eq.~\eqref{eq:LMMBound} is exact, which is observed in the coincidence with the direct optimization of the local unitaries (bottom curve). 
    As for the SDP approximation of Eq. \eqref{eq:SDP_upper_bound}
    under $0$-th level of the DPS hierarchy [Eq.\,\eqref{dps0}]
    (middle curve), it represents an improvement at least with respect to estimating the PE via global ergotropy (top curve), although it is insufficient to close the gap with the actual value.
    }
    \label{fig:algo_comparison_Werner}
\end{figure}

\begin{figure}[b]
\includegraphics[width=0.8\linewidth]{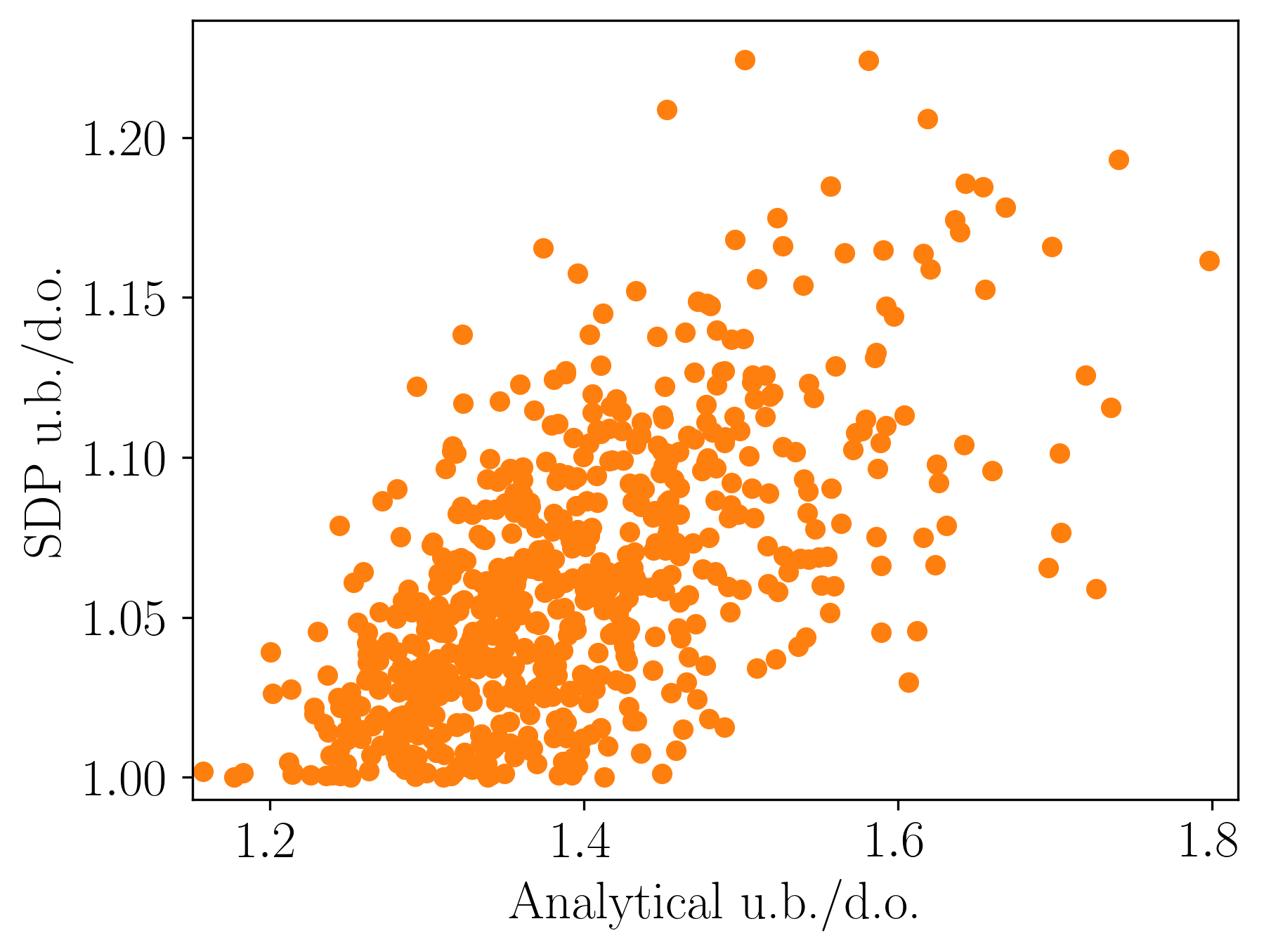}
    \caption{Comparison between the analytical upper bound (u.b.) of Eq. \eqref{eq:LMMBound} and the SDP upper bound of \eqref{sdp-bound-explicit} under $0$-th level of the DPS hierarchy. 
    Values are obtained for random two-qutrit states and random Hamiltonians with no local terms as detailed in Sec.\,\ref{sec:sdp-appl}. 
    Direct optimization of local unitaries is used to lower bound the actual PE. 
    Each point identifies a ratio between the two upper bounds and the lower bound.
    Abscissa:
    ratio between the the analytical u.b. and the direct optimization lower bound (d.o.).
    Ordinate:
    ratio between the SDP u.b. and the direct optimization lower bound. 
    Here the SDP technique provides an advantage in estimating PE in many cases, even sometimes attaining the exact value.}
    \label{fig:algo_comparison_d33}
\end{figure}

Applying the SDP technique \eqref{eq:SDP_upper_bound} to higher-dimensional systems introduces computational complications since the dimensions of $J_\Lambda$ grow as the square of the local dimensions. 
It is still possible, however, to perform tests for two-qutrit and two-ququart systems (still using just the  $0$-th level of the DPS hierarchy, Eq.\,\eqref{dps0}). 
In Fig.\,\ref{fig:algo_comparison_d33}, a comparison between the SDP approximation and the analytical bound of \eqref{eq:LMMBound} is presented for two-qutrit states in terms of the ratio that each method produces with respect to the same lower bound for PE, obtained with direct optimization of the local unitaries. 
For each point in the graph, a random state and a random Hamiltonian are produced and all estimates are computed on the produced pair. States are obtained as $\rho_{ab} = X\,X^\dagger/\operatorname{tr}[XX^\dagger]$, with entries on $X$ normally sampled \cite{Ginibre}, and Hamiltonians are produced as $H = v_{ji}\,\sigma^i_a \otimes \sigma^j_b$, with $v_{ji}$ normally distributed, $\sigma^i_{a/b}$ corresponding to a generalized (orthonormal) Pauli basis. 
For the particular sampling used, the SDP method performs better in many cases, even closing the gap with the lower bound for some points.

\section{Parallel capacity as an entanglement witness} 
\label{sec:ParallelCap}
Now we shift the focus to another quantity characterizing the potential of quantum battery to store work, which has been introduced in Section \ref{sec:preliminaries}, quantum battery capacity \cite{BatteryCap}. We analyze it under restriction to local unitary cycles, introducing thereby the quantity of \textit{parallel capacity}, and demonstrate its role as an entanglement witness. 
\subsection{Parallel capacity}
We start by formalizing the definition of parallel capacity. Similarly to the definition of PE \eqref{eq:parallel-ergo-N} of a multipartite quantum battery via restriction of work extracting unitaries to local ones, we define the parallel capacity of a $N$-partite quantum system equipped with a Hamiltonian $H$ as the maximum energy gap that we can obtain between two states which are \textit{connected} by parallel unitary operations,

\begin{equation}\label{eq:ParCap}
    \mathcal{C}_{\rm P}(\rho, H):= \max_{U_{\mathrm{P}} \in \mathcal{U}_{\mathrm{P}}^N} \Bigl( \tr  [U_{\mathrm{P}} \rho U_{\mathrm{P}}^\dag H]\Bigr) - \min_{\tilde{U}_{\mathrm{P}} \in \mathcal{U}_{\mathrm{P}}^N} \Bigl( \tr  [ \tilde{U}_{\mathrm{P}} \rho \tilde{U}_{\mathrm{P}}^\dag H] \Bigr),
\end{equation}
where $\mathcal{U}_{\rm P}^N$ denotes the set of {\it $N$-parallel unitaries}, i.e., unitaries of form $U_{\mathrm{P}} = \otimes_{i=1}^N U_i$, with $U_i$ acting on $i$-th subsystem. Notice that the parallel capacity \eqref{eq:ParCap} is bounded from below by the corresponding PE \eqref{eq:parallel-ergo-N}, because $ \max_{U_{\mathrm{P}} \in \mathcal{U}_{\mathrm{P}}^N}\;\; \tr  [ H  U_{\mathrm{P}} \rho U_{\mathrm{P}}^\dag] \geq \tr  [ H \rho]$. On the other hand, it is naturally bounded from above by the maximal energy gap $\Vert H \Vert_{\infty} := E_{\rm max} - E_{0}$, which is the difference of the highest $E_{\rm max}$ and lowest $E_{0}$ eigenvalues of $H$. Note that, with this definition, we are making a slight abuse of notation here and the symbol $\Vert . \Vert_\infty$ should be understood as the operator norm for the shifted operator $H$, where the ground state is placed with zero energy.
Therefore,
\begin{equation}\label{eq:BoundsParCap}
    \mathcal{E}_{\rm P}(\rho, H) \leq \mathcal{C}_{\rm P}(\rho, H) \leq \Vert H \Vert_{\infty}.
\end{equation}
\begin{widetext}
In what follows, we focus on bipartite interacting quantum batteries equipped with the Hamiltonian \eqref{hamilt} and described by the state $\rho_{ab} \in D(\mathcal{H}_a \otimes \mathcal{H}_b)$. In particular, for a class of bipartite setups characterized by a Hamiltonian without local terms and/or LMM states, which have been analyzed in Section \ref{sec:bounds} in the context of PE, parallel capacity is bounded from above by

\begin{equation}\label{eq:ParCapUpperBound}
     \mathcal{C}_{\rm P}(\rho_{ab}, H_{ab}) \leq \begin{cases}
      2((\vec \lambda^{|\mathbf{V}|}\cdot \vec \lambda^{|{\mathbf T}|}) - \lambda_0^{|\mathbf{V}|} \lambda_0^{|\mathbf{T}|}) & \text{if $d$ is even.} \\
      2 (\Vec{\lambda}^{|\mathbf{V}|} \cdot \Vec{\lambda}^{|\mathbf{T}|}) & \text{if $d$ is odd and $\operatorname{det}(-\mathbf{V}\mathbf{T}) \geq 0$,}\\
      2((\vec \lambda^{|\mathbf{V}|}\cdot \vec \lambda^{|{\mathbf T}|}) - 2\lambda_0^{|\mathbf{V}|} \lambda_0^{|\mathbf{T}|}) & \text{if $d$ is odd and $\operatorname{det}(-\mathbf{V}\mathbf{T}) < 0$,}
    \end{cases}
\end{equation}
\end{widetext}
where equality holds for $d=2$ (see Appendix \ref{app:sec:PCbound}).

Imai \textit{et al} \cite{WorkFluctuations} showed that it is possible to detect entanglement in a bipartite quantum battery statistically by probing a large variance of average work extracted under Haar-random parallel unitaries. 
We proceed by demonstrating an alternative entanglement criterion via detection of a large parallel capacity \eqref{eq:ParCap} of the working medium.

\subsection{Maximum parallel capacity of separable states and entanglement gap energy}

It is known that entanglement of a pure bipartite quantum state $\ket{\Psi_{ab}}$ can be captured by considering its Schmidt decomposition $\ket{\psi_{ab}} = \sum_{i=1}^{r_S(\Psi_{ab})} \sqrt{p_i} \ket{\phi_{i,a}} \otimes \ket{\psi_{i,b}}$, where $\sum_i p_i = 1$ and $\{ \ket{\phi_{i,a}}\}_i$ and $\{ \ket{\psi_{i,b}}\}_i$ are certain orthonormal bases in $\mathcal{H}_a$ and $\mathcal{H}_b$, respectively. In turn, the state $\ket{\Psi_{ab}}$ is entangled if and only if its Schmidt rank $r_S(\Psi_{ab}) > 1$. For a mixed state $\rho_{ab}$, the Schmidt rank can be generalized to Schmidt number that captures the pure states with the highest Schmidt rank from all its possible ensemble realizations $\rho_{ab} = \sum_i p_i |\Psi_{i, ab}\rangle\langle \Psi_{i, ab} |$ \cite{GenSN, WorkFluctuations},
\begin{equation}
    \tilde{r}_S(\rho_{ab}) = \operatorname{inf}\limits_{\mathcal{D}(\rho_{ab})} \operatorname{max}\limits_{\Psi_{i, ab}} r_S(\Psi_{i, ab}),
\end{equation}
where $\mathcal{D}(\rho_{ab}) = \{ p_i, |\Psi_{i, ab}\rangle\}$ is the set of ensemble realizations of $\rho_{ab}$.

Bipartite density matrices $\rho_{ab}$ with Schmidt number $\tilde{r}_S(\rho_{ab}) = k$ establish a hierarchy of subsets $\mathcal{S}^k \subset D(\mathcal{H}_a \otimes \mathcal{H}_b)$ that are convex and satisfy $ \mathcal{S}^{l}\subseteq \mathcal{S}^{k} \iff l\leq k$. Conversely, we say that $\rho_{ab}$ has Schmidt number $\tilde{r}_S(\rho_{ab})=k$ if $\rho_{ab} \in \mathcal{S}^{k} \backslash \mathcal{S}^{k-1}$. Notice that $ \mathcal{S}^1$ coincides with the set of separable states. Therefore, $\tilde{r}_S(\rho_{ab})> 1 $ implies that $\rho_{ab}$ is entangled, whereas higher Schmidt number $\tilde{r}_S(\rho_{ab})$ implies stronger entanglement in $\rho_{ab}$ \cite{GenSN0, GenSN1}. In order to connect this observation to parallel capacity of $\rho_{ab}$, we define the \textit{maximum parallel capacity} of states belonging to $\mathcal{S}^{k}$ as
\begin{equation}
    \mathcal{C}_{\rm P}^{k}(H_{ab}):=\max_{\rho_{ab} \in \mathcal{S}^{k}} \mathcal{C}_{\rm P}(\rho_{ab},H_{ab})\,.
\end{equation} 
Consequently, any state $\rho_{ab}$ such that $\mathcal{C}_{\rm P}(\rho_{ab}, H_{ab}) > \mathcal{C}_{\rm P}^{k}(H_{ab})$ necessarily has Schmidt number higher than $k$. This suggests that $\mathcal{C}_{\rm P}^{1}$ can be used to check whether a state has Schmidt number $k > 1$ and, hence, is entangled. In the following theorem, we show that it is indeed strictly connected to the entanglement gap energy \cite{EntaglementGap}, which is a known entanglement witness and defined as the energy difference between the separable state with minimum energy and the ground state energy,
\begin{equation}\label{eq:EntGap}
    \Delta E_{\mathrm{Sep}}(H_{ab}) = \min_{\rho \in \mathcal{S}^1} \operatorname{Tr}[\rho H_{ab}] - E_0,
\end{equation}
where $E_0$ is the energy of the ground state of $H_{ab}$.

\begin{theorem}\label{thm:SepCap}
Given a bipartite Hamiltonian $H_{ab}$, the corresponding maximum parallel capacity of separable states is given by
\begin{equation}
\label{CP1sep}
    \mathcal{C}_{\rm P}^{1}(H_{ab})=\Vert H_{ab}\Vert_{\infty}-(\Delta E_{\mathrm{Sep}}(H_{ab})+\Delta E_{\mathrm{Sep}}(-H_{ab})).
\end{equation}
\end{theorem}
\begin{proof}
First, it is immediate to verify that the parallel capacity \eqref{eq:ParCap} (hence, also the maximal parallel capacities $\mathcal{C}_{\rm P}^k(H_{ab})$) is convex in $\rho_{ab}$. Therefore, the maximum in $\mathcal{C}_{\rm P}^1(H_{ab})$ can be always achieved by a pure separable state, and the optimization can be performed over product states $\ket{\Psi} = \ket{\phi}_a \otimes \ket{\psi}_b$,
\begin{equation}
    \mathcal{C}_{\rm P}^{1}(H_{ab}) = \max_{\ket{\Psi} \in \mathcal{S}^1} \mathcal{C}_{\rm P}(\ket{\Psi}\bra{\Psi}, H_{ab})\,,
\end{equation}
On the other hand, for any pair of product states $\ket{\Psi}$ and $\ket{\tilde{\Psi}}$, it is possible to find local unitaries $U_a \otimes U_b$ that transform them into each other, so that $\ket{\tilde{\Psi}} = (U_a \otimes U_b) \ket{\Psi}$. Therefore, the maximization over local unitaries in the parallel capacity $\mathcal{C}_{\rm P}(\ket{\Psi}\bra{\Psi}, H_{ab})$ can be substituted by maximization over the set of pure separable states. Moreover, any product state $\ket{\Psi}$ has the same parallel capacity $\mathcal{C}_{\rm P}(\ket{\Psi}\bra{\Psi}, H_{ab})$, therefore,
\begin{equation}\label{eq:thm:ParCap1}
    \mathcal{C}_{\rm P}^1(H_{ab}) = \max_{|\Psi\rangle \in \mathcal{S}^1} \bra{\Psi}H_{ab}\ket{\Psi} - \min_{|\Psi\rangle \in \mathcal{S}^1} \bra{\Psi}H_{ab}\ket{\Psi}.
\end{equation} 
Taking into account that entanglement gap energy \eqref{eq:EntGap} can be always achieved by a pure separable state as well, we can connect the latter term in \eqref{eq:thm:ParCap1} to it as
\begin{equation}\label{eq:thm:minSep}
    \min_{\ket{\Psi} \in \mathcal{S}^1 } \bra{\Psi}H_{ab}\ket{\Psi} = \Delta E_{\mathrm{Sep}}(H_{ab}) + E_0.
\end{equation}
On the other hand, the maximization in \eqref{eq:thm:ParCap1} can be substituted by minimization with respect to the Hamiltonian with the opposite sign,
\begin{equation}
    \max_{\ket{\Psi} \in \mathcal{S}^1 }\; \bra{\Psi}H_{ab}\ket{\Psi} = - \min_{\ket{\Psi} \in \mathcal{S}^1 }\; \bra{\Psi}(-H_{ab})\ket{\Psi},
\end{equation}
and, hence, connected to the corresponding entanglement gap energy as
\begin{equation}\label{eq:thm:maxSep}
    \max_{\ket{\Psi} \in \mathcal{S}^1 }\; \bra{\Psi}H_{ab}\ket{\Psi} = - \Delta E_{\mathrm{Sep}}(-H_{ab}) + E_{\mathrm{max}},
\end{equation}
where it is taken into account that the ground state energy of $-H_{ab}$ is equal to $-E_{\mathrm{max}}$, the maximal eigenvalue of $H_{ab}$ taken with the opposite sign. Plugging in \eqref{eq:thm:minSep} and \eqref{eq:thm:maxSep} into \eqref{eq:thm:ParCap1}, we obtain \eqref{CP1sep}, hence, proving the thesis.

\end{proof}

We remark that there exist efficient numerical techniques to compute $\Delta E_{\mathrm{Sep}}$  \cite{EntaglementGap}, making the computation of $\mathcal{C}_{\rm P}^{1}(H_{ab})$ efficient too. 
As a corollary, we get that $\mathcal{C}_{\rm P}^{1}(H_{ab})=\Vert H_{ab}\Vert_{\infty}$
iff both the ground state energy manifold and the maximum energy manifold contain at least a separable state. 
\subsection{Entanglement witness} \label{EntaglementWitness}
Theorem \ref{thm:SepCap} reveals that parallel capacity of a state provides an entanglement witness under suitable choice of Hamiltonian $H_{ab}$. Performing optimization of the parallel capacity over all states of the battery,
\begin{equation}
   \mathcal{C}_{\rm P}(H_{ab}):= \max_{\rho_{ab} } \;\; \mathcal{C}_{\rm P}(\rho_{ab},H_{ab})\,,
\end{equation}
we formulate the following criterion.
\begin{observation}\label{Criterion}
Given a Hamiltonian $H_{ab}$ such that
\begin{equation}\label{eq:CondHamCrit}
    \mathcal{C}_{\rm P}(H_{ab}) > \Vert H_{ab}\Vert_{\infty}-\big (\Delta E_{\mathrm{Sep}}(H_{ab})+\Delta E_{\mathrm{Sep}}(-H_{ab})\big ),
\end{equation} 
fulfillment of condition
\begin{equation}\label{eq:ParCapCrit}
    \mathcal{C}_{\rm P}(H_{ab}, \rho_{ab}) > \Vert H_{ab}\Vert_{\infty}-\big (\Delta E_{\mathrm{Sep}}(H_{ab})+\Delta E_{\mathrm{Sep}}(-H_{ab})\big ),
\end{equation}
witnesses entanglement of $\rho_{ab}$.
\end{observation}

As a paradigmatic example, we consider the antiferromagnetic Hamiltonian \eqref{eq:AntiferrH} for 2 qubits, which can be given as
\begin{equation}
    H_{\mathrm{Ant}} = \omega(  \mathbb{1}_{ab}-4\ket{\psi^{-}}\bra{\psi^{-}}),
\end{equation}
where $\ket{\psi^{-}} := \frac{1}{\sqrt{2}}(\ket{01}-\ket{10})$ is the singlet state. First, let us show that $H_{\mathrm{Ant}}$ is a suitable Hamiltonian that fulfills \eqref{eq:CondHamCrit}. In order to calculate $\mathcal{C}_{\rm P}(H_{\mathrm{Ant}})$, we recall that, as a parallel capacity, it features the bounds \eqref{eq:BoundsParCap}, where the lower bound is true for any state $\rho_{ab}$. For the upper bound $\Vert H_{\mathrm{Ant}}\Vert_\infty$, we find that $\ket{\psi^{-}}$ is the ground state of $H_{\mathrm{Ant}}$ with associated energy $E_{0}=-3 \omega$, while the orthogonal subspace to $\ket{\psi^{-}} $ is the maximum energy manifold with associated energy $E_{\max}= \omega$. Therefore,
\begin{equation}
 \Vert H_{\mathrm{Ant}}\Vert_\infty = 4 \omega.
\end{equation}
On the other hand, taking $\rho_{ab} = \ket{\Psi^+}\bra{\Psi^+}$, it is straightforward to show that its ergotropy
\begin{equation}
\mathcal{E}_{\mathrm{P}}(\ket{\Psi^+}\bra{\Psi^+}, H_{\rm Ant}) = 4\omega,
\end{equation}
providing the lower bound on the maximal parallel capacity. Hence, we conclude that $\mathcal{C}_{\rm P}(H_{\mathrm{Ant}}) = 4\omega$. 
In order to check that it satisfies \eqref{eq:CondHamCrit}, we notice that the maximum energy manifold can be spanned by the basis $\{\ket{00}, \ket{11}, \ket{\Psi^+}\}$. This implies $\Delta E_{\mathrm{Sep}}(-H_{\mathrm{Ant}})=0$, because, e.g., the separable state $\ket{00}$ is a ground state for $-H_{\mathrm{Ant}}$. On the other hand, the antiferromagnetic Hamiltonian is known to provide the largest possible entanglement gap for 2 qubits given by \cite{EntaglementGap}
\begin{equation}
    \Delta E_{\mathrm{Sep}}(H_{\mathrm{Ant}}) = \Bigl(1 - \frac{1}{d}\Bigr) \Vert H_{\mathrm{Ant}}\Vert_\infty = 2  \omega.
\end{equation}
In turn, applying Theorem \ref{thm:SepCap}, we obtain $\mathcal{C}_{\rm P}^{1}(H_{\mathrm{Ant}}) = 2  \omega $, which is strictly smaller than $\mathcal{C}_{\rm P}(H_{\mathrm{Ant}})$. Therefore, $H_{\mathrm{Ant}}$ fulfills the condition \eqref{eq:CondHamCrit}, and the corresponding parallel capacity can be used to detect entanglement.

Now, let us consider 2-qubit Werner states \eqref{eq:WernerState}, which are known to be entangled iff $p > 1/3$. First, we take into account that the representation \eqref{eq:AntiferrH} of $H_{\mathrm{Ant}}$ suggests that $\mathbf{V}$ is proportional to identity, and $\Vec{\lambda}^{|\mathbf{V}|} = \omega \begin{pmatrix} 1 \\ 1 \\ 1 \end{pmatrix}$, while Werner states are described by $\Vec{\lambda}^{|\mathbf{T}|} = p \begin{pmatrix} 1 \\ 1 \\ 1 \end{pmatrix}$. Applying \eqref{app:eq:ParCapBound}, where the equality holds, since we consider qubit subsystems, we obtain
\begin{equation}
    \mathcal{C}_{\rm P}(W_p,H_{\mathrm{Ant}})=4 p  \omega\,.
\end{equation}
Therefore, the parallel capacity witnesses entanglement if and only if 
\begin{equation}\label{eq:CapCrit}
    p > \frac{1}{2}.    
\end{equation}

\subsection{Parallel capacity and work fluctuations}

While entanglement in quantum battery can be verified by its high capacity under local unitaries, extraction of work via randomly chosen local unitaries provides an alternative statistical entanglement witness. Indeed, entanglement in a bipartite $d \times d$ quantum battery can be verified by detecting large fluctuations in the corresponding work statistics \cite{WorkFluctuations}. Given a Hamiltonian $H_{ab}$, fulfillment of condition
\begin{eqnarray}\label{eq:WorkFluctCrit}
    \nonumber (\Delta \overline{W})^2 &>& \frac{1}{d^2 - 1} \Bigl( \frac{d^2}{4}( |\Vec{R}_{a}|^2 |\Vec{h}_{a}|^2 + |\Vec{R}_{b}|^2 |\Vec{h}_{b}|^2 ) \\
    &+& \frac{s_1(d, \Vec{R}_{a}, \Vec{R}_{b}) \Vert \mathbf{V} \Vert_2}{d^2 - 1} \Bigr),
\end{eqnarray}
where $s_1(d, \Vec{R}_{a}, \Vec{R}_{b}) = d - 1 + \frac{d^2}{4} \Bigl(\frac{d-2}{2} (|\Vec{R}_{a}|^2 + |\Vec{R}_{b}|^2 ) - \frac{d}{2} \Bigl| |\Vec{R}_{a}|^2 - |\Vec{R}_{b}|^2 \Bigr| \Bigr)$ and $\Vert \mathbf{V} \Vert_2 = \operatorname{tr}[\mathbf{V}^T \mathbf{V}] = \sum_{i=0}^{d^2-2} (\lambda_i^{|\mathbf{V}|})^2$, witnesses entanglement in $\rho_{ab}$. Taking into account decomposition of the work variance into generalized Pauli operators \cite{WorkFluctuations},
the criterion \eqref{eq:WorkFluctCrit} can be given in a more simple form, which depends only on the \textit{structure of the state} of quantum battery,
\begin{equation}\label{eq:WorkFluctCritSimple}
    \frac{d^4}{16} \Vert \mathbf{T} \Vert_2 > s_1(d, \Vec{R}_{a}, \Vec{R}_{b}).
\end{equation}
In order to compare it with the parallel capacity criterion \eqref{eq:ParCapCrit}, following Section \ref{sec:bounds}, we focus on a class of 2-qubit setups characterized by a Hamiltonian without local terms and/or locally maximally mixed states. In this case, for a given Hamiltonian $H_{ab}$, it is always possible to find states that do not fulfill \eqref{eq:WorkFluctCritSimple}, hence, do not lead to fluctuations of work high enough to detect entanglement via \eqref{eq:WorkFluctCrit}, yet have a parallel capacity that witnesses entanglement with respect to the criterion \eqref{eq:ParCapCrit} (see Appendix \ref{app:WorkFluct}).

As an example, let us turn back to the case of antiferromagnetic Hamiltonian \eqref{eq:AntiferrH} and 2-qubit Werner states \eqref{eq:WernerState}. For the latter, $s_1 = 1$, and the work fluctuation criterion \eqref{eq:WorkFluctCrit} witnesses their entanglement if and only if
\begin{equation}
    p > \frac{1}{\sqrt{3}}.
\end{equation}
Comparing it with the parallel capacity criterion \eqref{eq:CapCrit}, we conclude that there exists a subset of Werner states with $\frac{1}{2} < p \leq \frac{1}{\sqrt{3}}$ that are entangled due to their large parallel capacity despite low work fluctuations.

\section{Extended Parallel Ergotropy} \label{EPE}
We conclude our treatment by discussing the concept of extended parallel ergotropy (EPE).
In close analogy with the recently introduced extended local ergotropy (ELE)~\cite{ELE}, 
EPE is defined by setting as allowed operations free time evolutions and parallel unitaries.
The set of {\it extended parallel unitaries} is then
\begin{equation} \label{LocalOp}
    \mathcal{U}_{\rm Pex}:=\big \{
    U \;{\rm s.t.}\; U= \!\! \mathcal{T} \exp \big [-i\!\!\int_{0}^{T}\!\!\!\!H_{ab}+H_{a}(t)+H_{b}(t)dt \big ]\big \} 
    \,.
\end{equation}
Interestingly, a seminal result of quantum computing states that for any interacting Hamiltonian $H_{ab}$ this set includes all unitary operations on the joint $AB$ system \cite{SimHamDxD}. 

This means that $\mathcal{U}_{\rm Pex}$ coincides with the set of unitaries on the joint $AB$ system \cite{SimHam,SimHamDxD}, implying
\begin{equation}
    \forall \rho  , \;\mathcal{E}_{\rm Pex}(\rho,H)= \mathcal{E}(\rho,H)\,.
\end{equation} 
Notably, a series of parallel unitaries together with the entangling unitaries provided by internal time evolution enable any unitary transformation on the joint system. 
This comes as a result of quantum control theory \cite{Q.ControlBook,Q.Control(1972)} providing at least in some cases a practical way of implementing the desired unitary. 
More specifically, since $\mathtt{U}(d_{a}d_{b}) $ is compact, there exist a number $N$ and a time $T$ (both finite) such that any $ U \in \mathtt{U}(d_{a}d_{b})$
 can be implemented through a sequence of $N$ parallel unitaries and free time evolution,
with a total free evolution non greater than $T$. However, such $N $ and $T$ depend on the specific model and might be large \cite{Q.ControlBook, Q.Control(1972)}.

\section{Conclusions}
\label{sec:conc}
Maximum quantum work extraction is generally defined in terms of the \textit{ergotropy} functional, which does not capture potential difficulties in experimental implementation of the optimal unitary allowing for it, in particular, for multipartite systems.
In this framework, we consider a quantum battery consisting of many interacting sub-systems and study the maximum extractable work via concurrent local unitary operations on each subsystem. It is defined in terms of a functional that we call \textit{parallel ergotropy}, a quantity that in many cases is expected to be experimentally more accessible than the ergotropy of the entire battery. Thanks to its properties, we have derived useful bounds and exact analytical results for specific classes of states and/or Hamiltonians and provided receipts for numerical estimation of upper bounds via semi-definite programming techniques in the generic case.
In particular, we provided exact expressions of parallel ergotropy of qubit bipartite batteries in a locally maximally mixed state and equipped with an arbitrary Hamiltonian or in arbitrary state and equipped with a locally null Hamiltonian. 
We have also observed that parallel ergotropy outperforms work extraction via \textit{egoistic} strategies, in which the first agent $A$ extracts locally the maximum work, followed by the second agent $B$ that extracts locally the remaining work. This indicates that cooperation between the agents in work extraction is necessary for an overall benefit. 

Apart from ergotropy, we have studied another quantity characterizing capability of a quantum system to accumulate and supply work, \textit{quantum battery capacity}. Restriction to the concurrent local unitary operations has revealed that its counterpart, \textit{parallel capacity}, can serve as a witness of entanglement. We have compared it with another entanglement witness based on statistical properties of work extraction from a multipartite quantum battery, namely revealing entanglement under large variance of the average work with respect to random Haar parallel unitaries. We found that a high parallel capacity that detects entanglement does not necessarily imply fluctuations of work high enough to detect entanglement via its large variance.
Finally, we showed that if free time evolutions and parallel unitaries are allowed, one can always extract the maximum work stored in the quantum battery, i.e., saturate the gap between the parallel ergotropy and ergotropy of the entire system.

Our work leaves several open research problems. Although we introduced the parallel ergotropy and quantum battery capacity for a system composed of an arbitrary number $N$ of subsystems, our study has been focused mainly on the case $N=2$. 
 Therefore, analysis of these quantities beyond the bipartite case would be a natural development of the present work. In particular, their relation to multipartite entanglement and behavior in the limit $N \rightarrow \infty$ have to be examined. 

\paragraph*{Acknoledgements.---}
We thank Antonio Ac\'in and Vittorio Giovannetti for fruitful discussions.
DF acknowledges financial support from PNRR MUR Project No. PE0000023-NQSTI. RN acknowledges support from the Government of Spain (Severo Ochoa CEX2019-000910-S and TRANQI), Fundació Cellex, Fundació Mir-Puig, Generalitat de Catalunya (CERCA program) and support from the Quantera project Veriqtas. This research was funded in whole or in part by the Austrian Science Fund (FWF) 10.55776/PAT4559623. For open access purposes, the author has applied a CC BY public copyright license to any author-accepted manuscript version arising from this submission.

\begin{widetext}

\appendix

\section{Generalized Pauli operator  formalism for parallel ergotropy} \label{GPO}
In order to facilitate calculation of PE \eqref{eq:parallel-ergo} for bipartite quantum batteries of arbitrary dimension, we treat the corresponding states and Hamiltonians exploiting the generalized Pauli operators (GPO) expansion formalism \cite{GOP}. Given a quantum system, let $d$ be the dimension of the underlying Hilbert space. A GPO set is then a collection of $ d^2-1 $ operators $ \sigma^{j} $ such that 
\begin{equation}
    \tr [\sigma^{j} ]=0 , \;\;\; 
    \tr [\sigma^{i}\sigma^{j} ]=2\delta_{ij}.
\end{equation}
For $ d=2 $, this set can be reduced to standard Pauli operators, while for standard form of the GPO set for $ d>3 $ we refer the reader to \cite{GOP}. Equipped with the identity operator, GPO form an orthogonal basis on the real vector space of self-adjoint operators acting on the $d$-dimensional Hilbert space.
 
Let us consider a bipartite $(d_a \times d_b)$-dimensional quantum system in state $\rho_{ab} \in D(\mathcal{H}_a \otimes \mathcal{H}_b)$ and equipped with the Hamiltonian \eqref{hamilt}, where, without loss of generality, we assume
\begin{equation}
    \tr_a[V_{ab}] = \tr_b[V_{ab}] = 0.
\end{equation}
Then we can represent the state and the local and non-local components of the Hamiltonian in the following way using the GPO,
\begin{eqnarray}
\label{StateDef}
\rho_{ab} &=& \frac{\mathbb{1}_{ab}}{d_ad_b} + \frac{1}{2}\Bigl((\vec R_a \cdot \vec \sigma_a) \otimes \frac{\mathbb{1}_b}{d_b}\Bigr) + \frac{1}{2}\Bigl( \frac{\mathbb{1}_a}{d_a} \otimes (\vec R_b \cdot \vec \sigma_b)\Bigr) + \frac{1}{4}\sum_{i=1}^{d_a^2-1} \sum_{j=1}^{d_b^2-1} {\mathbf T}_{i,j} (\sigma_{a}^i \otimes \sigma_{b}^j),\\
H_a &=& (\vec h_a \cdot \vec \sigma_a  ) \otimes \mathbb{1}_b
, \;\;\;\;\;\;\;\; H_{ b} \; = \; \mathbb{1}_a \otimes (\vec h_b \cdot \vec \sigma_b  ), \;\;\;\;\;\;\; V_{ab} \; = \; \sum_{i=1}^{d_a^2-1} \sum_{j=1}^{d_b^2-1} \mathbf{V}_{j,i} (\sigma_{a}^i \otimes \sigma_{b}^j),\label{HamiltonianDef}
\end{eqnarray}
where
\begin{eqnarray}
\vec R_a &=& \tr\Bigl[(\vec \sigma_{a} \otimes \mathbb{1}_b) \rho_{ab}\Bigr], \;\;\;\;\;\;\;\; \vec R_b \; = \; \tr\Bigl[(\mathbb{1}_a \otimes \vec \sigma_{b}) \rho_{ab}\Bigr], \;\;\;\;\;\;\; {\mathbf T}_{i,j} \; = \; \tr[(\sigma_{a}^i \otimes \sigma_{b}^j)\rho_{ab}],
\end{eqnarray}
and
\begin{eqnarray}
\vec h_a &=& \frac{1}{2d_b}\tr\Bigl[(\vec \sigma_{a} \otimes \mathbb{1}_b) H_{ab}\Bigr], \;\;\;\;\;\;\;\; \vec h_b \; = \; \frac{1}{2d_a} \tr\Bigl[(\mathbb{1}_a \otimes \vec \sigma_{b}) H_{ab}\Bigr], \;\;\;\;\;\;\;\; \mathbf{V}_{j,i} \; = \; \frac{1}{4}\tr[V_{ab}(\sigma_{a}^i \otimes \sigma_{b}^j)].
\end{eqnarray}
Notice the swapped indices in the definition $\mathbf{V}_{j,i} $, so that introduction of transpose of $\mathbf{V}$ into the energy functional can be omitted,
\begin{eqnarray} \label{Energy}  \nonumber E(\rho_{ab}) &:=& \tr[H_{ab}\rho_{ab}] \\
&=& (\vec h_a \cdot \vec R_a) + (\vec h_b \cdot \vec R_b) + \tr[\mathbf{V}{\mathbf T}] \,.
\end{eqnarray}
GPO expansion of $\rho_{ab}$ and $H_{ab}$ in \eqref{StateDef} and \eqref{HamiltonianDef} suggests the following simplification of optimization over unitaries in the PE functional \eqref{eq:parallel-ergo}. Since the adjoint representation of $\mathtt{SU}(n)$ (hence, the representation of $\mathtt{U}(n)$ on $\mathfrak{su}(n)$ as well) is a subgroup $\mathcal{SO}(n^2-1) \subseteq \mathtt{SO}(n^2-1)$ \cite{Wolf1968}, for any unitaries $U_{a} \in \mathtt{U}(d_a)$, $U_{b} \in \mathtt{U}(d_b)$ acting on $ \mathcal{H}_{a}$ and $\mathcal{H}_{b}$, respectively, there exists a pair of associated rotations $\mathbf{O}_a \in \mathtt{SO}(d_{a}^2-1)$ and $\mathbf{O}_b \in \mathtt{SO}(d_{b}^2-1) $ transforming $\vec{R}_{a/b}$ and $\mathbf{T}$,
\begin{equation}
(\mathrm{U}_a \otimes \mathrm{U}_b)[\rho_{ab}] = \frac{\mathbb{1}_{ab}}{d_ad_b} + \frac{1}{2}\Bigl((\mathbf{O}_a\vec R_a \cdot \vec \sigma_a) \otimes \frac{\mathbb{1}_b}{d_b}\Bigr) + \frac{1}{2}\Bigl( \frac{\mathbb{1}_a}{d_a} \otimes (\mathbf{O}_b \vec R_b \cdot \vec \sigma_b)\Bigr) + \frac{1}{4}\sum_{i=1}^{d_a^2-1} \sum_{j=1}^{d_b^2-1} \mathbf{O}_a {\mathbf T}_{i,j} \mathbf{O}_b^T (\sigma_{a}^i \otimes \sigma_{b}^j).
\end{equation}
Therefore, the optimization over $\mathtt{U}(d_a)$ and $\mathtt{U}(d_b)$ in \eqref{eq:parallel-ergo} can be substituted by optimization over $\mathcal{SO}(d_a^2-1)$ and $\mathcal{SO}(d_b^2-1)$, respectively,
\begin{equation} 
    \label{parallel-pauli}
    \mathcal{E}_{\rm P}(\rho_{ab}, H_{ab}) = E(\rho_{ab})-\underset{\substack{\mathbf{O}_a \in  \mathcal{SO}(d_{a}^2-1) \\ \mathbf{O}_b \in  \mathcal{SO}(d_{b}^2-1)}}{\rm min}\;\; \Bigl( (\vec h_a \cdot \mathbf{O}_a \vec R_a) + (\vec h_b\cdot \mathbf{O}_b \vec R_b) + \tr[\mathbf{V}\mathbf{O}_a{\mathbf T}\mathbf{O}_b^{T}]  \Bigr)
\end{equation}
Since $\mathcal{SO}(n^2-1)$ is a proper subgroup of $\mathtt{SO}(n^2-1)$ (unless $n = 2$, when they coincide), we can establish upper bound on PE by performing the optimization over the entire rotation group,
\begin{equation} 
    \mathcal{E}_{\rm P}(\rho_{ab}, H_{ab}) \leq  E(\rho_{ab})-\underset{\substack{\mathbf{O}_a \in  \mathtt{SO}(d_{a}^2-1) \\ \mathbf{O}_b \in  \mathtt{SO}(d_{b}^2-1)}}{\rm min}\;\; \Bigr( (\vec h_a \cdot \mathbf{O}_a \vec R_a) + (\vec h_b\cdot \mathbf{O}_b \vec R_b) + \tr[\mathbf{V}\mathbf{O}_a{\mathbf T}\mathbf{O}_b^{T}]  \Bigl),
\label{GenericUpper}
\end{equation}
which is saturated for $d_a = d_b = 2$, i.e., bipartite qubit systems.

\section{Egoistic strategies for quasi-generic Hamiltonians}
\label{app:egoistic}

Let us consider a bipartite quantum battery composed of two qubits. Using the decomposition of its state \eqref{StateDef} and Hamiltonian \eqref{HamiltonianDef} via Pauli operators, we assume that the latter is characterized by $\mathbf{V}$ with 3 distinct eigenvalues and $\operatorname{det}[\mathbf{V}] > 0$. In turn, the state $\rho_{ab}$ is chosen to satisfy the following conditions,
\begin{eqnarray} \label{app:eq:condego}
    \mathbf{T} &=& -\eta\mathbb{1}, \\
    \Vec{R}_a &=& 0, \\
    (\Vec{R}_b \cdot \vec h_b) &\neq& -\frac{|\vec R_b|}{|\vec h_b|}, \\
    |\Vec{R}_b| &=& 1/2. \label{app:eq:condego4}
\end{eqnarray}
Before to proceed with comparison of cooperative and egoistic strategies for this setup, let us verify of existence of physical states defined by these conditions. The positivity and normalization of a bipartite qubit state can be guaranteed by the following criterion \cite{2-Qubits},
\begin{align}
& 3-\Bigl(\lVert \mathbf{T} \rVert^2_2 + |\vec R_a|^2 + |\vec R_b|^2\Bigr) \geq 0, \\
& 2\Bigl(\vec R_a^T \mathbf{T} \vec R_b - \operatorname{det}[\mathbf{T}]\Bigr) - \Bigl(\lVert \mathbf{T} \rVert^2_2 + |\vec R_a|^2 + |\vec R_b|^2 - 1 \Bigr) \geq 0,  \\
& 8\Bigl(\vec R_a^T (\mathbf{T} + \overline{\mathbf{T}}) \vec R_b - \operatorname{det}[\mathbf{T}]\Bigr) + \Bigl(\lVert \mathbf{T} \rVert^2_2 + |\vec R_a|^2 + |\vec R_b|^2 - 1\Bigr)^2 - 4\Bigl(| \vec R_a|^2 |\vec R_b|^2 + \lVert \mathbf{T} \rVert^2_2 (|\vec R_a|^2 + |\vec R_b|^2) + \lVert \overline{\mathbf{T}} \rVert^2_2\Bigr) \geq 0, 
\end{align}
where $\lVert X \rVert_2 = \sqrt{\operatorname{Tr}[X^T X]}$, and $\overline{\mathbf{T}}$ is the cofactor matrix of $\mathbf{T}$ defined by the identity $\mathbf{T} \overline{\mathbf{T}}^T = \overline{\mathbf{T}}^T \mathbf{T} = \operatorname{det}[\mathbf{T}]\mathbb{1}$. The chosen state $\rho_{ab}$ is straightforwardly characterized by $\operatorname{det}[\mathbf{T}] = -\eta^3$ and $\lVert \mathbf{T} \rVert_2 = \sqrt{3} \eta$. Therefore, the corresponding cofactor matrix is given by $\overline{\mathbf{T}} = \eta^2 \mathbb{1}$ with $\lVert \overline{\mathbf{T}} \rVert_2 = \sqrt{3} \eta^2$. Therefore, the criterion is reduced to
\begin{eqnarray}
    \eta^2 &\leq& \frac{11}{12}, \\
    \eta^2 (3 - 2\eta) &\leq& \frac{3}{4}, \\
    \eta^2 (15 - 16\eta + 6\eta^2) &\leq& \frac{9}{8}.
\end{eqnarray}
The last inequality appears to be the most restrictive one, and the entire set of conditions is fulfilled if $ \eta\in [-|\eta_{\rm min}|, \eta_{\rm max}]$, where $\eta_{\rm min} \approx -0.241$ and $\eta_{\rm max} \approx 0.329$, proving the existence of physical states in this range of $\eta$.

Assuming $\eta > 0$, first, we calculate PE for the system satisfying \eqref{app:eq:condego}--\eqref{app:eq:condego4},
\begin{equation}\label{app:eq:par-ergo}
       \mathcal{E}_{\rm P}(\rho_{ab}, H_{ab}) = (\vec h_b \cdot \vec R_b) + \eta \operatorname{tr}[(-\mathbf{V})] - \min_{\mathbf{O}_a, \mathbf{O}_b \in \mathtt{SO}(3)}\Bigl( (\vec h_b\cdot \mathbf{O}_b \vec R_b ) + \eta \;\tr[(-\mathbf{V})\mathbf{O}_a\mathbf{O}_b^T] \Bigr),
\end{equation}
where we have taken into account that the upper bound \eqref{GenericUpper} is saturated for a bipartite quantum battery. Since minimization over both $\mathbf{O}_a$ and $\mathbf{O}_b$ in \eqref{app:eq:par-ergo} is performed simultaneously, we can substitute it by minimization over $\mathbf{O}_1 := \mathbf{O}_a$ and $\mathbf{O}_2 := \mathbf{O}_a \mathbf{O}_b^T$,
\begin{equation}\label{app:eq:par-ergo-min}
       \mathcal{E}_{\rm P}(\rho_{ab}, H_{ab}) = (\vec h_b \cdot \vec R_b) + \eta \operatorname{tr}[(-\mathbf{V})] - \Bigl( \min_{\mathbf{O}_1 \in \mathtt{SO}(3)} (\vec h_b\cdot \mathbf{O}_b \vec R_b ) + \eta \min_{\mathbf{O}_2 \in \mathtt{SO}(3)} \tr[(-\mathbf{V})\mathbf{O}_2] \Bigr),
\end{equation}
Straightforwardly, the optimal $\overline{\mathbf{O}}_1$ is such that $\overline{\mathbf{O}}_1 \vec R_b = -\frac{|R_b|}{|h_b|} \vec h_b $. The second minimization is achieved by a unique (since $\mathbf{V}$ is non-degenerate) optimal $\overline{\mathbf{O}}_2$ such that $\operatorname{tr}[(-\mathbf{V})\overline{\mathbf{O}}_2] = -\operatorname{tr}[|\mathbf{V}|] = -\sum_{i=0}^{d^2-1} \lambda_i^{|\mathbf{V}|}$, where $\{\lambda_i^{|\mathbf{V}|}\}_i$ are singular values of $\mathbf{V}$, i.e., eigenvalues of $|\mathbf{V}| = \sqrt{\mathbf{V}^T \mathbf{V}}$. Therefore, we obtain the following PE,
\begin{equation}
     \mathcal{E}_{\rm P}(\rho_{ab},H_{ab}) = (\vec h_b \cdot \vec R_b) + |h_b||R_b| + \eta \tr\bigl[|\mathbf{V}| - \mathbf{V}\bigr].
\end{equation}
Notice that the local unitary operation performed by $A$ may require energy. In this case, some energy has to be pumped by $A$ into the system, so that $B$ is able to extract the maximal work globally.

Now let us calculate the ergotropy $\mathcal{E}_{\rm La}(\rho, H) + \mathcal{E}_{\rm Lb}(\overline{\mathsf{U}}_a[\rho], H)$ under egoistic strategies, as introduced in the hierarchy \eqref{CooperationGap}. The maximal work that can be extracted by the agent $A$ is given by
\begin{eqnarray}
    \nonumber \mathcal{E}_{\rm La}(\rho_{ab}, H_{ab}) &=& (\vec h_b \cdot \vec R_b) + \eta \operatorname{tr}[(-\mathbf{V})] - \min_{\mathbf{O}_a \in \mathtt{SO}(3)}\Bigl( (\vec h_b\cdot \vec R_b ) + \eta \;\tr[(-\mathbf{V})\mathbf{O}_a] \Bigr) \\
    &=& \eta \;\tr\bigl[|\mathbf{V}| - \mathbf{V}\bigr],
\end{eqnarray}
where we have taken into account that the optimal unitary $\overline{U}_a$ of $A$ is associated to $\overline{\mathbf{O}}_2$. In turn, the maximal work that $B$ can extract after $A$ is
\begin{eqnarray}
    \nonumber \mathcal{E}_{\rm Lb}(\overline{\mathsf{U}}_a[\rho_{ab}], H_{ab}) &=& (\vec h_b \cdot \vec R_b) + \eta \operatorname{tr}[(-\mathbf{V})\overline{\mathbf{O}}_2] - \min_{\mathbf{O}_b \in \mathtt{SO}(3)}\Bigl( (\vec h_b\cdot \mathbf{O}_b \vec R_b ) + \eta \;\tr[(-\mathbf{V})\overline{\mathbf{O}}_2 \mathbf{O}_b^T] \Bigr) \\
    &=& (\vec h_b \cdot \vec R_b) - \eta \operatorname{tr}[|\mathbf{V}|] - \min_{\mathbf{O}_b \in \mathtt{SO}(3)}\Bigl( (\vec h_b\cdot \mathbf{O}_b \vec R_b ) - \eta \;\tr[|\mathbf{V}| \mathbf{O}_b^T] \Bigr).\label{app:eq:min-ego}
\end{eqnarray}
Minimization of the last term in \eqref{app:eq:min-ego} requires $\mathbf{O}_b = \mathbb{1}$. On the other hand, in order to minimize the term $(\vec h_b\cdot \mathbf{O}_b \vec R_b )$ it is necessary to apply a unitary associated to $\overline{\mathbf{O}}_1$. Therefore, the optimal values of both terms cannot be achieved simultaneously, and the maximum work that can extracted by $B$ is strictly upper-bounded by
\begin{eqnarray}
    \nonumber \mathcal{E}_{\rm Lb}(\overline{\mathsf{U}}_a[\rho_{ab}], H_{ab}) < (\vec h_b \cdot \vec R_b) - \eta \operatorname{tr}[|\mathbf{V}|] + |\vec h_b||\vec R_b| + \eta \;\tr[|\mathbf{V}|] = (\vec h_b \cdot \vec R_b) + |\vec h_b||\vec R_b|,
\end{eqnarray}
and we obtain a strict inequality for PE and the ergotropy under egoistic strategies with respect to $A$,
\begin{equation}
    \mathcal{E}_{\rm La}(\rho_{ab}, H_{ab}) + \mathcal{E}_{\rm Lb}(\overline{\mathsf{U}}_a[\rho_{ab}], H_{ab}) < \mathcal{E}_{\rm P}(\rho_{ab}, H_{ab}).
\end{equation}
Following the same steps, it is straightforward to demonstrate that egoistic work extracting strategies in the swapped order of the agents does not allow to achieve PE as well. Indeed, if agent $B$ acts first, it cannot find a unitary that minimizes both last terms simultaneously since $\overline{\mathbf{O}}_1 \neq \overline{\mathbf{O}}_2$. Therefore, it can extract maximal work
\begin{eqnarray}
    \nonumber \mathcal{E}_{\rm Lb}(\rho_{ab}, H_{ab}) &=& (\vec h_b \cdot \vec R_b) + \eta \operatorname{tr}[(-\mathbf{V})] - \min_{\mathbf{O}_b \in \mathtt{SO}(3)}\Bigl( (\vec h_b\cdot \mathbf{O}_b \vec R_b ) + \eta \;\tr[(-\mathbf{V})\mathbf{O}_b^T] \Bigr) \\
    &=& (\vec h_b \cdot \vec R_b) + \eta \operatorname{tr}[(-\mathbf{V})] - (\vec h_b\cdot \overline{\mathbf{O}}_b\vec R_b ) + \eta \;\tr[\mathbf{V}\overline{\mathbf{O}}_b^T],
\end{eqnarray}
where $\overline{\mathbf{O}}_b$ is associated to the optimal unitary. After that, $A$ can extract maximal work
 \begin{eqnarray}
     \nonumber \mathcal{E}_{\rm La}(\overline{\mathsf{U}}_b[\rho_{ab}], H_{ab}) &=& (\vec h_b \cdot \overline{\mathbf{O}}_b\vec R_b) + \eta \operatorname{tr}[(-\mathbf{V})\overline{\mathbf{O}}_b^T] - \min_{\mathbf{O}_a \in \mathtt{SO}(3)}\Bigl( (\vec h_b\cdot \mathbf{O}_b \vec R_b ) + \eta \;\tr[(-\mathbf{V})\mathbf{O}_a\overline{\mathbf{O}}_b^T] \Bigr) \\
    &=& \eta \operatorname{tr}\bigl[|\mathbf{V}| - \mathbf{V}\overline{\mathbf{O}}_b^T\bigr].
 \end{eqnarray}
Summing both contributions up and taking into account that $\overline{\mathbf{O}}_b \neq \overline{\mathbf{O}}_1$, we find that PE is a strict upper bound in this case as well,
\begin{equation}
    \mathcal{E}_{\rm Lb}(\rho_{ab}, H_{ab})+\mathcal{E}_{\rm La}(\overline{\mathsf{U}}_b[\rho_{ab}], H_{ab}) = (\vec h_b \cdot \vec R_b) + (\vec h_b\cdot \overline{\mathbf{O}}_b\vec R_b ) + \eta \tr\bigl[|\mathbf{V}| - \mathbf{V}\bigr] < \mathcal{E}_{\rm P}(\rho_{ab}, H_{ab}).
\end{equation}

For $\operatorname{det}[\mathbf{V}] < 0$, choosing $\eta < 0$ leads to the same results. Therefore, we conclude that except for particular Hamiltonians characterized by $\operatorname{det}[\mathbf{V}]=0$ or degenerate spectrum, there exist states that open the gap between cooperative and egoistic work extracting strategies.

\section{Upper bound on parallel ergotropy in specific cases} \label{ProofLLMBound}
Considering bipartite systems with $d_a = d_b = d$, we assume that one of hypotheses
\begin{enumerate}[label=(\roman*)]
    \item $\Vec{R}_{a} = \Vec{R}_{b} = 0 \;\;\; \Rightarrow \;\;\; \tr_a [\rho_{ab}] = \frac{\mathbb{1}_b}{d}, \; \tr_b [\rho_{ab}] = \frac{\mathbb{1}_a}{d}$,
    \item $\Vec{h}_{a} = \Vec{h}_{b} = 0 \;\;\;\; \Rightarrow \;\;\; H_{a} = H_{b} = 0$.
\end{enumerate}
Hypothesis (i) refers to the assumption of a bipartite system being in a locally maximally mixed state, while the Hamiltonian is left generic. Hypothesis (ii), instead, assumes a specific non-local Hamiltonian (e.g., the antiferromagnetic Hamiltonian \eqref{eq:AntiferrH}), while the state is left generic. 
Although requiring both hypotheses to be true simultaneously is natural from the mathematical point of view, it lacks physical relevance. Therefore, in what follows, we require only a single hypothesis to be true. This cancels out the local terms in \eqref{GenericUpper}, and the PE functional is upper-bounded as follows,
\begin{eqnarray}
\label{opt-interac}
    \nonumber \mathcal{E}_{\mathrm{P}}(\rho_{ab}, H_{ab}) &\leq& E(\rho_{ab}) - \underset{\mathbf{O}_a,\mathbf{O}_b \in \mathtt{SO}(d^2 - 1)}{\rm min} \Bigl( \tr[\mathbf{V}\mathbf{O}_a\mathbf{T}\mathbf{O}_b^{T}] \Bigr) \\
    &=& E(\rho_{ab}) + \underset{\mathbf{O}_a,\mathbf{O}_b \in \mathtt{SO}(d^2 - 1)}{\rm max} \Bigl( \tr[(-\mathbf{V})\mathbf{O}_a\mathbf{T}\mathbf{O}_b^{T}] \Bigr).
\end{eqnarray}
Therein, it is necessary to distinguish between two cases depending on signs of eigenvalues of $-\mathbf{V}\mathbf{T}$.

\subsection{Case $ \det(-\mathbf{V}\mathbf{T}) \geq 0 $}
In this case, either both $-\mathbf{V}$ and $\mathbf{T} $ have an \textit{even} number of negative eigenvalues or both have an \textit{odd} number of them. In the first case, 
we can decompose both matrices as $-\mathbf{V}=\mathbf{O_{\mathbf{V}}^{'}\mathbf{\Sigma}_{\mathbf V}O_{V}}$ and  $\mathbf{T}=\mathbf{O_{{\mathbf T}}^{'}\mathbf{\Sigma}_{\mathbf T}O_{{\mathbf T}} }$, where $\mathbf{O_V}, \mathbf{O_V'}, \mathbf {O_T}, \mathbf{O_T'} \in \mathtt{SO}(d^2 - 1)$, and $\mathbf{\Sigma_V}$ and $\mathbf{ \Sigma_T}$ are diagonal, positive, and having eigenvalues of $|\mathbf{V}|$ and $|\mathbf{T}|$, respectively, in non-decreasing order. 
In the second case, we use a similar decomposition, however, the value of the smallest eigenvalue is left negative, so that $\mathbf{\Sigma_V}={\rm diag}(-\lambda^{|\mathbf{V}|}_{0}, \lambda^{|\mathbf{V}|}_{1}, \dots, \lambda^{|\mathbf{V}|}_{d^{2}-2}) $ and $\mathbf{\Sigma_T}={\rm diag}(-\lambda^{|{\mathbf T}|}_{0}, \lambda^{{|\mathbf T|}}_{1},\dots,\lambda^{|{\mathbf T}|}_{d^{2}-2}) $. In both cases, the optimization \eqref{opt-interac} reduces to
\begin{equation}
    \mathcal{E}_{\mathrm{P}}(\rho_{ab}, H_{ab}) \leq E(\rho_{ab}) + \underset{\mathbf{O}_1,\mathbf{O}_2 \in \mathtt{SO}(d^2 - 1)}{\max} \tr[\mathbf{\Sigma}_{\mathbf V}\mathbf{O}_1\mathbf{\Sigma}_{\mathbf T}\mathbf{O}_2]\,,
\end{equation}
where $\mathbf{O}_1 = \mathbf{O_V} \mathbf{O}_a \mathbf{O_T'}$ and $\mathbf{O}_2 = \mathbf{O_T} \mathbf{O}_b \mathbf{O_V'}$. Finally, the Von Neumann trace inequality \cite{Rearr} allows us to provide the following upper bound for any $\mathbf{O}_1, \mathbf{O}_2 \in \mathtt{SO}(d^2 - 1)$, 
\begin{equation}   \tr[\mathbf{\Sigma}_{\mathbf V}\mathbf{O}_1\mathbf{\Sigma}_{\mathbf T}\mathbf{O}_2] \leq (\vec \lambda^{|\mathbf{V}|}\cdot \vec \lambda^{|{\mathbf T}|}),
\end{equation}
where $\vec\lambda^{|\mathbf{V}|}$ and $\vec\lambda^{|\mathbf{T}|}$ are the vectors composed by $\{\lambda_i^{|\mathbf{V}|}\}_i$ and $\{\lambda_i^{|\mathbf{T}|}\}_i$, respectively. Therefore,
\begin{eqnarray}\label{eq:app:ErgBoundPlus}
   \mathcal{E}_{\rm P}(\rho_{ab}, H_{ab}) \leq \tr[\rho_{ab} V_{ab}] + (\vec \lambda^{|\mathbf{V}|}\cdot \vec \lambda^{|{\mathbf T}|})\,,
   \quad
   {\rm if}\,\,\det(-\mathbf{V}\mathbf{T})\geq 0\,. 
\end{eqnarray}
\subsection{Case $\det(-\mathbf{V}\mathbf{T})<0$}
In this case, the proof proceeds in the same way as above. Nevertheless, it is necessary to take into account that the decomposition of $\mathbf{V}$ and $\mathbf{T}$ via $\mathtt{SO}(d^{2}-1)$-matrices results in diagonal matrices $\mathbf{\Sigma_V}$ and $\mathbf{\Sigma_T}$ that have an odd and even number of negative eigenvalues, respectively, or vice versa. This means that, in their product not all negative signs cancel out, and the maximum is achieved thereby by leaving the smallest eigenvalue of $|\mathbf{\Sigma}_{\mathbf V}|$ (or $|\mathbf{\Sigma}_{\mathbf T}|$) negative in $\mathbf{\Sigma}_{\mathbf{V}}$ (or $\mathbf{\Sigma}_{\mathbf{T}}$), leading to
\begin{equation}\label{eq:app:ErgBoundMinus}
   \mathcal{E}_{\rm P}(\rho_{ab},H_{ab}) \leq \tr[\rho_{ab} V_{ab}] + \Bigl( (\vec \lambda^{|\mathbf{V}|}\cdot \vec \lambda^{|{\mathbf T}|}) - 2\lambda^{|\mathbf{V}|}_{0} \lambda^{|{\mathbf T}|}_{0}\Bigr), \quad
   {\rm if}\,\,\det(-\mathbf{V}\mathbf{T}) < 0\,,
\end{equation}
and concluding the proof of inequalities (\ref{eq:LMMBound}).

\section{Upper bound on parallel capacity in specific cases} \label{app:sec:PCbound}

Similarly to parallel ergotropy, parallel capacity can be decomposed as follows,
\begin{align} 
\label{parallel-cap-pauli}
    \nonumber \mathcal{C}_{\rm P}(\rho_{ab}, H_{ab}) 
&= \underset{\mathbf{O}_a,\;\mathbf{O}_b \in  \mathcal{SO}(d_{a,b}^2-1)}{\rm max}\;\; \big( \vec h_a \cdot \mathbf{O}_a \vec R_a +\vec h_b\cdot \mathbf{O}_b \vec R_b + \tr[\mathbf{V}\mathbf{O}_a{\mathbf T}\mathbf{O}_b^{T}]  \big) \\
&- \underset{\mathbf{O}_a,\;\mathbf{O}_b \in  \mathcal{SO}(d_{a,b}^2-1)}{\rm min}\;\; \big( \vec h_a \cdot \mathbf{O}_a \vec R_a +\vec h_b\cdot \mathbf{O}_b \vec R_b + \tr[\mathbf{V}\mathbf{O}_a{\mathbf T}\mathbf{O}_b^{T}]  \big)\,,
\end{align}
and, by enlarging the domain from $\mathcal{SO}(d^2-1)$ to $\mathtt{SO}(d^2-1)$, bounded from above,
\begin{align} 
    \nonumber\mathcal{C}_{\rm P}(\rho_{ab}, H_{ab})  &\leq \underset{\mathbf{O}_a,\;\mathbf{O}_b \in  \mathtt{SO}(d_{a,b}^2-1)}{\rm max}\;\; \big( \vec h_a \cdot \mathbf{O}_a \vec R_a +\vec h_b\cdot \mathbf{O}_b \vec R_b + \tr[\mathbf{V}\mathbf{O}_a{\mathbf T}\mathbf{O}_b^{T}]  \big) \\
    &- \underset{\mathbf{O}_a,\;\mathbf{O}_b \in  \mathtt{SO}(d_{a,b}^2-1)}{\rm min}\;\; \big( \vec h_a \cdot \mathbf{O}_a \vec R_a +\vec h_b\cdot \mathbf{O}_b \vec R_b + \tr[\mathbf{V}\mathbf{O}_a{\mathbf T}\mathbf{O}_b^{T}]  \big)\,,
\label{GenericCapUpper}
\end{align}
where the equality holds iff $d=2$. Applying hypothesis (i) or (ii) of Appendix \ref{ProofLLMBound} to the state and Hamiltonian, 
\begin{equation}
\label{opt-interac-cap}
     \mathcal{C}_{\rm P}(\rho_{ab}, H_{ab})  \leq \underset{\mathbf{O}_a,\;\mathbf{O}_b \in \mathtt{SO}(d^2 - 1)}{\rm max}\;\; \tr[\mathbf{V}\mathbf{O}_a\mathbf{T}\mathbf{O}_b^{T}]\, + \underset{\mathbf{O}_a,\;\mathbf{O}_b \in \mathtt{SO}(d^2 - 1)}{\rm max}\;\; \tr[(-\mathbf{V})\mathbf{O}_a\mathbf{T}\mathbf{O}_b^{T}]\, ,
\end{equation}
Taking into account the linearity of the trace, so that $\operatorname{tr}(-X) = -\operatorname{tr}(X)$, it is enough to consider optimization of one of the terms in \eqref{opt-interac-cap}. Without loss of generality, we perform the optimization of the second term that coincides with \eqref{opt-interac}, hence, is given by \eqref{eq:app:ErgBoundPlus} for $\operatorname{det}(-\mathbf{V}\mathbf{T}) \geq 0$ and \eqref{eq:app:ErgBoundMinus} $\operatorname{det}(-\mathbf{V}\mathbf{T}) < 0$. Then, we take into account multilinearity of determinant of a matrix, so that $\operatorname{det}(-X) = (-1)^{n}\operatorname{det}(X)$, where $n=d^2 - 1$ is dimension of $\mathbf{V}$ and $\mathbf{T}$. Therefore, for even $d$, the optimization of the first term in \eqref{opt-interac-cap} will be given by \eqref{eq:app:ErgBoundPlus} for $\operatorname{det}(-\mathbf{V}\mathbf{T}) < 0$ and \eqref{eq:app:ErgBoundMinus} $\operatorname{det}(-\mathbf{V}\mathbf{T}) \geq 0$. This leads to
\begin{eqnarray}
    \nonumber \mathcal{C}_{\rm P}(\rho_{ab}, H_{ab})  &\leq& ( \Vec{\lambda}^{|\mathbf{V}|} \cdot \Vec{\lambda}^{|\mathbf{T}|}) + (\Vec{\lambda}^{|\mathbf{V}|} \cdot \Vec{\lambda}^{|\mathbf{T}|}) - 2 \lambda_0^{|\mathbf{V}|} \lambda_0^{|\mathbf{T}|}) \\
    &=& 2\Bigl( (\Vec{\lambda}^{|\mathbf{V}|} \cdot \Vec{\lambda}^{|\mathbf{T}|}) - \lambda_0^{|\mathbf{V}|} \lambda_0^{|\mathbf{T}|}  \Bigr). \label{app:eq:ParCapBound}
\end{eqnarray}
On the other hand, for odd $d$, the optimization of the first term in \eqref{opt-interac-cap} will be the same as for the second one, so that 
\begin{equation}
     \mathcal{C}_{\rm P}(\rho_{ab}, H_{ab}) \leq \begin{cases}
      2 (\Vec{\lambda}^{|\mathbf{V}|} \cdot \Vec{\lambda}^{|\mathbf{T}|}), & \text{if $\operatorname{det}(-\mathbf{V}\mathbf{T}) \geq 0$,}\\
      2((\vec \lambda^{|\mathbf{V}|}\cdot \vec \lambda^{|{\mathbf T}|}) - 2\lambda_0^{|\mathbf{V}|}\lambda_0^{|\mathbf{T}|}) & \text{if $\operatorname{det}(-\mathbf{V}\mathbf{T}) < 0$.}
    \end{cases}
\end{equation}

\section{Comparison of entanglement criteria in specific cases for two qubits}
\label{app:WorkFluct}
First, we write down the parallel capacity criterion, taking into account the saturated upper bound \eqref{eq:ParCapUpperBound} for $d=2$,
\begin{equation}
    2\Bigl(( \Vec{\lambda}^{|\mathbf{V}|} \cdot \Vec{\lambda}^{|\mathbf{T}|}) - \lambda_0^{|\mathbf{V}|} \lambda_0^{|\mathbf{T}|}  \Bigr) > \mathcal{C}_1(H_{ab}),
\end{equation}
where we have taken into account decomposition \eqref{app:eq:ParCapBound} in terms of eigenvalues of $\mathbf{V}$ and $\mathbf{T}$. In order to decompose in the same way the maximum parallel capacity of separable states, we recall that it is convex, so that it is enough to optimize the parallel capacity over pure separable states,
\begin{equation}
    \mathcal{C}_{\rm P}^{1}(H_{ab}) = \max_{\ket{\Psi} \in \mathcal{S}^1} \mathcal{C}_{\rm P}(\ket{\Psi}\bra{\Psi}, H_{ab})\,.
\end{equation}
Since $\ket{\Psi} = \ket{\phi_a} \otimes \ket{\psi_b}$ are product states, the corresponding matrix $\mathbf{T}$ is a Kronecker product of the corresponding Bloch vectors $\Vec{n}_{\phi_a}$ and $\Vec{n}_{\psi_b}$. In turn, such a matrix has a single non-zero eigenvalue, which is equal to the scalar product $(\Vec{n}_{\phi_a} \cdot \Vec{n}_{\psi_b})$. Therefore, only a single term in the eigenvalue decomposition \eqref{app:eq:ParCapBound} survives, and the optimal state corresponds to the case when the highest eigenvalue $\lambda_2^{|\mathbf{V}|}$ of $\mathbf{V}$ survives,
\begin{equation}
    \mathcal{C}_{\rm P}^{1}(H_{ab}) = 2 \lambda_2^{|\mathbf{V}|} \max_{\Vec{n}_{\phi_a}, \Vec{n}_{\psi_b}} (\Vec{n}_{\phi_a} \cdot \Vec{n}_{\psi_b}).
\end{equation}
Since the optimization is performed over pure states, so that both Bloch vectors have magnitude $1$, the maximum is achieved for $\Vec{n}_{\phi_a} = \Vec{n}_{\psi_b}$, and we obtain the following form of the parallel capacity criterion,
\begin{equation}
    \lambda_1^{|\mathbf{V}|} \lambda_1^{|\mathbf{T}|} + \lambda_2^{|\mathbf{V}|} \lambda_2^{|\mathbf{T}|} > \lambda_2^{|\mathbf{V}|}.
\end{equation}

On the other hand, the work fluctuation criterion for $d=2$ requires
\begin{equation}
    (\lambda_0^{|\mathbf{T}|})^2 + (\lambda_1^{|\mathbf{T}|})^2 + (\lambda_2^{|\mathbf{T}|})^2 > 1 - \Bigl| |\Vec{R}_a|^2 - |\Vec{R}_b|^2 \Bigr|^2.
\end{equation}
We question whether there exist entangled states that can be detected by the parallel capacity criterion but do not lead to large variance of stochastic work distribution. Inverting the previous inequality and taking into account that $\lambda_0^{|\mathbf{T}|} \leq \lambda_1^{|\mathbf{T}|} \leq \lambda_2^{|\mathbf{T}|}$, we find that entanglement in the states satisfying
\begin{equation}
    \Biggl(\lambda_0^{|\mathbf{T}|} < \sqrt{\frac{1 - \Bigl| |\Vec{R}_a|^2 - |\Vec{R}_b|^2 \Bigr|^2}{3}} \leq \frac{1}{\sqrt{3}} \Biggr) \cup \Biggl( \lambda_2^{|\mathbf{T}|} > \frac{\lambda_2^{|\mathbf{V}|}}{\lambda_1^{|\mathbf{V}|} + \lambda_2^{|\mathbf{V}|}} \geq \frac{1}{2} \Biggr)
\end{equation}
can be detected by parallel capacity, yet not by the work fluctuations criterion.

\section{Constraints for local unitality and approximate separability}
\label{app:approx_constraints}

The constraints that establish the relevant properties for $J_\Lambda$ as an approximation to a product of unitary channels are given by the last two lines in Eq.~\eqref{eq:SDP_upper_bound}. In what follows, we prove that the condition for local unitality [Eq.~\eqref{eq:def_local_unital}] in terms of its Choi operator is equivalent to 
\begin{equation}
    \mathrm{Tr}_{\mathtt{in}^{(i)}}[J_\Lambda] = \frac{1}{d_i}\mathrm{Tr}_{\mathtt{in}^{(i)},\,\mathtt{out}^{(i)}}[J_\Lambda]\otimes \mathbb{1}_{\mathtt{out}^{(i)}}.
    \label{eq:app_loc_unitality_Choi}
\end{equation}
Further below, we comment on the remaining constraint, which approximates separability of the channel $\Lambda$.

We restrict our discussion to the bipartite scenario to simplify notation, but the result can be straightforwardly adapted to the $n$-partite case.

\begin{lemma}
\label{lemma-loc-unital}
Eq.~\eqref{eq:def_local_unital} is equivalent to Eq.~\eqref{eq:app_loc_unitality_Choi}
\end{lemma}

\begin{proof}
First, note that, in terms of Choi states, local unitality [Eq.~\eqref{eq:def_local_unital}] is equivalent to
\begin{equation}
\mathrm{Tr}_{\mathtt{in}}\left[J_\Lambda \,\left(\rho^T \otimes \mathbb{1}_{\mathtt{in}^{(i)}} \otimes \mathbb{1}_{\mathtt{out}}\right)\right] = \rho' \otimes \mathbb{1}_{\mathtt{out}^{(i)}}.
\label{eq:loc_unital_direct}
\end{equation}
By tracing subsystem $\mathtt{out}^{(i)}$ on both sides of the equation above, it can be deduced that $\rho' = \frac{1}{d_i}\,\mathrm{Tr}_{\mathtt{in},\,\mathtt{out}^{(i)}}[J_\Lambda \,\left(\rho^T \otimes \mathbb{1}_{\mathtt{in}^{(i)}}\otimes\mathbb{1}_{\mathtt{out}}\right)]$, and the condition for local unitality can be equivalently rewritten as
\begin{align}
\mathrm{Tr}_{\mathtt{in}}&\left[J_\Lambda \left(\rho^T \otimes \mathbb{1}_{\mathtt{in}^{(i)}} \otimes \mathbb{1}_{\mathtt{out}}\right)\right] = \nonumber \\ 
&\frac{1}{d_i}\,\mathrm{Tr}_{\mathtt{in},\,\mathtt{out}^{(i)}}\left[J_\Lambda \,\left(\rho^T \otimes \mathbb{1}_{\mathtt{in}^{(i)}}\otimes\mathbb{1}_{\mathtt{out}}\right)\right] \otimes \mathbb{1}_{\mathtt{out}^{(i)}}.
\label{eq:loc_unital_direct_2}
\end{align}\\

Assume all spaces involved to be finite-dimensional and let $\{B^1_{\mu_1} \otimes B^2_{\mu_2}\}_{\mu_1,\mu_2}$ be a orthogonal basis for operators acting on the input spaces, where $B^i_0 = \mathbb{1}_{\mathtt{in}^{(i)}}$ and $\mathrm{Tr}[B^{i\,\dag}_\mu\,B^{i}_\nu ] = d_i\,\delta_{\mu,\nu}$. Let then $J_\Lambda$ be decomposed as
\begin{equation}
    J_\Lambda = \sum_{\mu_1,\mu_2} B^1_{\mu_1} \otimes B^2_{\mu_2} \otimes A_{\mu_1,\mu_2},
    \label{eq:choi_basis_decomp}
\end{equation}
where $A_{\mu_1,\mu_2}$ are the remaining part of the decomposition, acting on the output spaces. Extend the action of $J_\Lambda$ to any bounded operator on the input spaces and use the basis elements (transposed) as inputs for the channel $\Lambda$, so that, e.g., for $\rho = \left(B^{1}_{\mu_1}\right)^T$, it is obtained
\begin{equation}
\frac{1}{d_1\,d_2}\mathrm{Tr}_{\mathtt{in}}\left[J_\Lambda \,\left(B^1_{\mu_1} \otimes \mathbb{1}_{\mathtt{in}^{(2)}} \otimes \mathbb{1}_{\mathtt{out}}\right)\right] = A_{\mu_1,0},
\label{eq:action_basis}
\end{equation}
which, combined with Eq. \eqref{eq:loc_unital_direct_2}, implies
\begin{equation}
A_{\mu_1,0} = \frac{1}{d_2}\,\mathrm{Tr}_{\mathtt{out}^{(2)}}\left[ A_{\mu_1,0} \right] \otimes \mathbb{1}_{\mathtt{out}^{(2)}}.    
\end{equation}
These are equations relating the basis components of $\mathrm{Tr}_{\mathtt{in}^{(2)}}[J_\Lambda]$, adjoining the elements $B^1_{\mu_1}$ in the corresponding input space and summing over the free index, the basis-independent expression is found
\begin{equation}
\mathrm{Tr}_{\mathtt{in}^{(2)}}[J_\Lambda] = \frac{1}{d_2}\,\mathrm{Tr}_{\mathtt{in}^{(2)},\,\mathtt{out}^{(2)}}\left[ J_{\Lambda} \right] \otimes \mathbb{1}_{\mathtt{out}^{(2)}}.    
\end{equation}
A similar reasoning can be applied for the identity on subsystem $1$ and Eq. \eqref{eq:app_loc_unitality_Choi} is obtained.\\

On the reverse direction, assuming Eq. \eqref{eq:app_loc_unitality_Choi} valid, it suffices to multiply the resulting operators on both sides by $\rho^T_{S-S_i}$ (meaning a state that acts on all subsystems except for $\mathtt{in}^{(i)}$). Use that $\mathrm{Tr}_{\mathtt{in}^{(i)}}[J_\Lambda]\cdot\rho^T_{S-S_i} = \mathrm{Tr}_{\mathtt{in}^{(i)}}[J_\Lambda\,(\rho^T_{S-S_i} \otimes \mathbb{1}_{\mathtt{in}^{(i)}})]$, and trace out the remaining input subspaces. The resulting relation is precisely Eq. \eqref{eq:loc_unital_direct_2}, which is the local unitality condition.
\end{proof}

\subsection{Some formal clarifications}
\label{clarifications}
Referring to Eq.\,\eqref{conv-def},
we notice that
because of linearity, we have 
\begin{equation}
\label{opt-pure}
     \min_{\mathbf{\Phi}\in \mathcal{C}} \tr  \Bigl[  \mathbf{\Phi}(\rho) H\Bigr]= \min_{\mathbf{\Phi}\in \textit{Conv}(\mathcal{C})} \tr  \Bigl[ \mathbf{\Phi}(\rho) H\Bigr]
     \,.
\end{equation}
This is enough to show, for instance, that the 
relaxation from 
$
\otimes_{i=1}^{n}\mathcal{U}^{\rm ni}_i
$ to $
Conv(\otimes_{i=1}^{n}\mathcal{U}^{\rm ni}_i)$
   in Eq.\,\eqref{inclusions-sets} comes at no cost.
Furthermore, the argument that achieves the maximization of the rhs can always be chosen such that is part of the $\mathcal{C}$ set. 

We now denote with $\mathcal{C}_{A},\mathcal{C}_{B}$ the two sets of channels that act locally on the A, B subsystems respectively.
\begin{lemma}
\label{lemma-convconv}
We claim that
\begin{equation}
    \min_{
        \mathbf{\Phi}_{A}\in \mathcal{C}_{A},
        \mathbf{\Phi}_{B}\in \mathcal{C}_{B}}
    \tr  \Bigl[ \mathbf{\Phi}_{A}\otimes \mathbf{\Phi}_{B}(\rho) H\Bigr]= 
    \min_{\mathbf{\Phi}_{A}\in \textit{Conv}(\mathcal{C}_{A}),\mathbf{\Phi}_{B}\in \textit{Conv}(\mathcal{C}_{B})} \tr  \Bigl[\mathbf{\Phi}_{A}\otimes \mathbf{\Phi}_{B}(\rho) H\Bigr] \,.
    \label{ConvConv}
\end{equation}
\end{lemma}
     
This is not trivially implied by the previous lemma \eqref{opt-pure}, since $\textit{Conv}(\mathcal{C}_{A})\otimes \textit{Conv}(\mathcal{C}_{B})\subseteq \textit{Conv}(\mathcal{C}_{A}\otimes \mathcal{C}_{B}) $, however it can be proven to hold true.
\paragraph*{Proof.---}
Let $\overline{\mathbf{\Phi}}_{A},\overline{\mathbf{\Phi}}_{B} $ be the pair of channels which achieve the minimum of the lhs of Eq.\,(\ref{ConvConv}). Then we have 
\begin{align}
   & \min_{\mathbf{\Phi}_{A}\in \textit{Conv}(\mathcal{C}_{A}),\mathbf{\Phi}_{B}\in \textit{Conv}(\mathcal{C}_{B})} \tr  \Bigl[\mathbf{\Phi}_{A}\otimes \mathbf{\Phi}_{B}(\rho) H\Bigr] \leq   \min_{\mathbf{\Phi}_{A}\in \textit{Conv}(\mathcal{C}_{A})} \tr  \Bigl[\mathbf{\Phi}_{A}\otimes \overline{\mathbf{\Phi}}_{B}(\rho) H\Bigr]= \\
   &\min_{\mathbf{\Phi}_{A}\in \mathcal{C}_{A}} \tr  \Bigl[\mathbf{\Phi}_{A}\otimes \overline{\mathbf{\Phi}}_{B}(\rho) H\Bigr]=  \min_{\mathbf{\Phi}_{A}\in \mathcal{C}_{A},\mathbf{\Phi}_{B}\in \mathcal{C}_{B}} \tr  \Bigl[ \mathbf{\Phi}_{A}\otimes \mathbf{\Phi}_{B}(\rho) H\Bigr].
\end{align}
On the other hand, for \eqref{opt-pure} we know that 
\begin{equation}
    \min_{\mathbf{\Phi}_{A}\in \mathcal{C}_{A},\mathbf{\Phi}_{B}\in \mathcal{C}_{B}} \tr  \Bigl[ \mathbf{\Phi}_{A}\otimes\mathbf{\Phi}_{B}(\rho) H\Bigr]=
    \min_{\mathbf{\Phi} \in \textit{Conv}(\mathcal{C}_{A} \otimes \mathcal{C}_{B})} \tr  \Bigl[ \mathbf{\Phi}(\rho) H\Bigr]\,. 
\end{equation}
The latter is smaller or equal than
\begin{equation}
    \min_{\mathbf{\Phi}\in Conv(\mathcal{C}_{A}) \otimes Conv(\mathcal{C}_{B})} \tr  \Bigl[ \mathbf{\Phi}(\rho) H\Bigr] 
\end{equation}
because $\textit{Conv}(\mathcal{C}_{A})\otimes \textit{Conv}(\mathcal{C}_{B})\subseteq \textit{Conv}(\mathcal{C}_{A}\otimes \mathcal{C}_{B})$, 
implying
\begin{equation}
    \min_{\mathbf{\Phi}\in Conv(\mathcal{C}_{A}) \otimes Conv(\mathcal{C}_{B})} \tr  \Bigl[ \mathbf{\Phi}(\rho) H\Bigr]
    \geq 
    \min_{\mathbf{\Phi}_{A}\in \mathcal{C}_{A},\mathbf{\Phi}_{B}\in \mathcal{C}_{B}} \tr  \Bigl[ \mathbf{\Phi}_{A}\otimes\mathbf{\Phi}_{B}(\rho) H\Bigr]\,,
\end{equation}
hence 
concluding the proof of Eq.\,\eqref{ConvConv}. 

Extending the reasoning to $n$ subsystems, this explicitly shows that
the relaxation from 
$\otimes_{i=1}^{n}\mathcal{U}_i $ to $\otimes_{i=1}^{n}Conv(\mathcal{U}_i)$
in Eq.\,\eqref{inclusions-sets} comes at no cost too.

\subsection{Approximate separability}

The remaining constraint in Eq. \eqref{eq:SDP_upper_bound} specifies that $J_\Lambda$ should be approximately fully separable, considering each party to be composed by the pair input and output subsystems. To simplify the discussion, a bipartite system is considered with subsystems $A$ and $B$. Extending the results to channels that act on more subsystems is immediate.

A separable channel $\Lambda$ is a channel with the structure
\begin{equation}
    \Lambda(\bullet) = \sum_i p_i\,\Lambda^i_A \otimes \Lambda^i_B(\bullet),
    \label{eq:app_sep_channel}
\end{equation}
where $p_i$ are the coefficients of the convex combination of the product channels $\Lambda^i_A \otimes \Lambda^i_B$, such that $\Lambda^i_A \otimes \Lambda^i_B (\rho_A \otimes \rho_B) = \Lambda^i_A(\rho_A) \otimes \Lambda^i_B(\rho_B)$. Channel separability immediately translates into separability of the corresponding Choi state across the different pairs of input and output subsystems:
\begin{align}
    J_\Lambda &= \sum p_i\,|u_1,\,u_2\rangle\langle v_1,\,v_2| \otimes \Lambda^i_A(|u_1\rangle\langle v_1|)\otimes \Lambda^i_B(|u_2\rangle\langle v_2|) \nonumber \\
    &= \sum_i p_i \,J_{\Lambda^i_A} \otimes J_{\Lambda^i_B},
\end{align}
where the sum in the first equality is done on all variables $i, u_1,u_2,v_1,v_2$. On the second equality, $J_{\Lambda^i_X}$ are the Choi states for each local channel ($X=A,B$): $J_{\Lambda^i_X} = \sum_{u,v} |u\rangle\langle v| \otimes \Lambda^i_X(|u\rangle\langle v|)$.

Therefore, a test for separability on $J_\Lambda$ similar to tests for separability on quantum states must be applied to ensure the ideal structure for the channel $\Lambda$. Since product states are extremal points in the set of separable states, it is expected that an approximation that converges to the set of separable channels should also converge to a pure product of unital channels. The hierarchy of tests introduced in \cite{DPS_2004} simplifies the problem of separability to linear restrictions on extended states, that represent a replication of the information used to produce the original state, a feat that can only be faithfully performed in an indefinite manner for separable states \cite{Caves_2002}.

Each level of the hirearchy includes the previous ones, with the ``level zero" representing postivity of partial transpose (PPT), meaning that not only $J_\Lambda$ must be positive semidefinite, but also the transposes over $\mathtt{in}^{(i)}, \mathtt{out}^{(i)}$, for any subsystem $i$. This is so because a full transpose of a channel preserves its positive-semidefiniteness, consequently, the partial transpose of a separable Choi state should correspond to the transposition of a local channel, mapping a valid separable channel into a different valid separable channel. Failure in complying to this criterion implies that the channel cannot be taken as separable.

Let $\mathsf{DPS}_k$ correspond to the set of Choi states that pass the $k$-th level of the test hierarchy, level $0$ corresponding to operators $J_\Lambda$ with PPT. For $k \geq 1$, the set $\mathsf{DPS}_k$ is also characterized by the operators that possess a $k$-symmetric extension in all parties, i.e. operators $J_\Lambda$ that admit a Choi state $J^{(k)}_\Lambda \geq 0$ acting on the extended space $\mathcal{H}^{\otimes(k+1)}_A \otimes \mathcal{H}_B$, such that
\begin{align}
    \mathsf{SWAP}_{a,b}\,\,J^{(k)}_\Lambda\,\,\mathsf{SWAP}_{a,b} &= J^{(k)}_\Lambda,\,\, \forall a,b \in \{0,\ldots,k\},\,a<b \nonumber \\
    \left(J^{(k)}_\Lambda\right)^{T_{0\ldots j}} &\geq 0,\,\, \forall j \in \{0,\ldots,k\}, \nonumber \\
    \mathrm{Tr}_{A_1\ldots A_k}[J^{(k)}_\Lambda] &= J_\Lambda,
    \label{eq:app_DPSk}
\end{align}
where $\mathsf{SWAP}_{a,b}$ is the swap operator for swapping subsytems $A_a$ and $A_b$ (indexed from $0$ to $k$) and $T_{0\ldots j}$ is partial transposition over all $j+1$ first copies of the subsystem $A$. 

The hierarchical structure of the test can be seen from Eq.~\eqref{eq:app_DPSk}: a state that admits a $(k+1)$-symmetric extension necessarily admits a $k$-symmetric extension, since tracing $J^{(k+1)}_\Lambda$ over any of the copy subsystems results in a Choi state that satisfies all the required properties for the $k$-th extension. Furthermore, a state is separable if and only if it passes the tests for all $k \geq 0$ \cite{Caves_2002, DPS_2004}, thereby ensuring completeness to the test. Any channel that is not exactly separable will fail the test for some given $k$. 

Although finding this certification could be costly if higher levels of the hierarchy become necessary, in practice it is observed that the $0$-th level test already bounds well the separable set for bipartite, low-dimensional systems. With this, a convenient approximation to separable channels is attained.

\section{Details on lower bounding PE and comparison between upper bounds}
\label{app:DirectUnitary}

For the results in section \ref{sec:bounds}, a lower bound on PE was computed by direct optimization of the local unitaries applied on each subsystem. For qubits, a convenient parametrization is given by
\begin{equation}
    U_{\theta,\phi,\gamma} = \begin{bmatrix}
        e^{i(\gamma-\phi)/2}\cos(\theta/2) & e^{-i(\gamma+\phi)/2}\sin(\theta/2) \\
        e^{i(\gamma+\phi)/2}\sin(\theta/2) & -e^{-i(\gamma-\phi)/2}\cos(\theta/2)
    \end{bmatrix},
\end{equation}
with $\theta \in [0,\pi]$ and $\phi, \gamma \in [0, 2\pi]$. This parametrization has been used for optimization using the Werner states (see Fig. \ref{fig:algo_comparison_Werner} in the main text), where direct nonlinear optimization using the Nelder-Mead algorithm was performed over $U_{\theta,\phi,\gamma}$. Given the symmetries of Werner states [Eq. \eqref{W-invariance}], optimization only over subsystem A was considered, although, given the small number of parameters, good convergence can also be obtained in more general cases, with optimization over both parties. In Fig. \ref{fig:algo_comparison_Werner}, convergence is observed, as the lower bound from direct optimization matches the upper bound given by the analytical formula of Eq. \eqref{eq:LMMBound}. For higher dimensions, an alternative parametrization in terms of generators of $SU(N)$ has been used. For a given orthonormal basis of Hermitian $N \times N$ matrices, $\{\sigma^j\}_{1 \leq j \leq d^2-1}$, with traceless $\sigma^j$, unitaries are given by $U = \exp(-i\,\sum_j x_j\,\sigma^j)$, with $x_j \in [0,\pi]$.  Lower bounds for double qutrit systems (Fig. \ref{fig:algo_comparison_d33}) were obtained using the parametrization in terms of generators. Closure of the gap between estimates indicates that the approximation is effective for given pairs of states and Hamiltonians produced.

\end{widetext}

\end{document}